\documentclass[
 reprint,
superscriptaddress,
 amsmath,amssymb,
 aps,
pra,
]{revtex4-2}

\usepackage{lineno}
\usepackage{braket}
\usepackage{graphicx}
\usepackage{dcolumn}
\usepackage{bm}
\usepackage{hyperref}
\usepackage{soul}
\usepackage{braket}
\usepackage{bbding}
\usepackage{multirow} 
\usepackage{mathrsfs}
\usepackage{cuted} 
\usepackage{tikz} 
\usepackage{xcolor}

\usetikzlibrary{matrix,fit}

\soulregister{\cite}{7}
\soulregister{\ref}{7}
\usepackage{amsmath}
\usepackage{amsthm}
\newtheorem{theorem}{Theorem}[section] 
 
\newtheorem{lemma}{Lemma} 

\newtheorem{proposition}[theorem]{Proposition}
\usepackage{balance}

\hypersetup{colorlinks=true, citecolor=blue, urlcolor=blue, linkcolor=blue}

\begin{document}

\preprint{APS/123-QED}
\title{Discrete-Modulated Continuous-Variable Quantum Key Distribution \\ with Uncertainty Principle}

\author{Jiale Mi}
\thanks{These authors contributed equally to this work.}
 \affiliation{%
 Beijing Key Laboratory of Optical Quantum Information Network, State Key Laboratory of Information Photonics and Optical Communications, School of Electronic Engineering, Beijing University of Posts and Telecommunications, Beijing, 100876, China
}%
\author{Yiming Bian}
\thanks{These authors contributed equally to this work.}
 \affiliation{%
 Beijing Key Laboratory of Optical Quantum Information Network, State Key Laboratory of Information Photonics and Optical Communications, School of Electronic Engineering, Beijing University of Posts and Telecommunications, Beijing, 100876, China
}%

\author{Song Yu}%
 \affiliation{%
 Beijing Key Laboratory of Optical Quantum Information Network, State Key Laboratory of Information Photonics and Optical Communications, School of Electronic Engineering, Beijing University of Posts and Telecommunications, Beijing, 100876, China
}%
\author{Zhengyu Li}%
 \email{lizhengyu2@huawei.com}
 \affiliation{%
 Huawei Technologies Co., Ltd., Shenzhen 518129, China
}%
\author{Yichen Zhang}%
 \email{zhangyc@bupt.edu.cn}
 \affiliation{%
 Beijing Key Laboratory of Optical Quantum Information Network, State Key Laboratory of Information Photonics and Optical Communications, School of Electronic Engineering, Beijing University of Posts and Telecommunications, Beijing, 100876, China
}%
\author{Hong Guo}%
 \email{hongguo@pku.edu.cn}
 \affiliation{%
 Beijing Key Laboratory of Quantum Sensing and Precision Measurement, and Center for Quantum Information \\ Technology, and
 Institute of Quantum Electronics, Peking University, Beijing 100871, China
}%

\date{\today}

\begin{abstract}
    Continuous-variable quantum key distribution is a compelling framework for scalable quantum networks due to its seamless integration with existing optical communication infrastructure. 
    However, a fundamental gap persists between theoretical protocols requiring ideal Gaussian modulation and the constrained, discrete-modulated signals dictated by practical high-speed hardware. 
    Current security proofs for discrete modulation rely on semidefinite programming, which suffers from prohibitive computational overhead for high-order constellations and lacks direct physical insight into non-Gaussian modulation.
    In this Letter, we overcome this limitation by developing a security framework that obviates semidefinite programming in favor of an approach grounded fundamentally in the Heisenberg uncertainty principle. By introducing a multi-mode entanglement-source model to characterize non-Gaussian state preparation, we establish an explicit mapping between constellation geometry and the secret key rate. This framework effectively quantifies the security implications of hardware-limited, finite state preparation, enabling both numerical and analytical security analysis under high-order constellations. We experimentally validate our method on both discrete-component and integrated photonic platforms, demonstrating that a quadrature amplitude modulation format with 256 constellation points can asymptotically approach the Gaussian capacity limit. 
    Beyond quantum key distribution, the principle of tightening uncertainty-constrained bounds via source-mode expansion offers a paradigm for exploring the information-theoretic properties of complex non-Gaussian systems.

\end{abstract}

\maketitle

Quantum key distribution (QKD) utilizes quantum mechanical principles to provide information-theoretical security for two remote parties \cite{bennet1984quantum, AdvInQC, PTPQKDRMV2020, portmann2022security}. Continuous-variable (CV) QKD protocols with coherent measurement exhibits high compatibility with standard telecommunications \cite{weedbrook2012gaussian,Zhang2024Continuous,usenko2026continuous}. There are two typical modulation schemes for CV protocols: Gaussian modulation \cite{ralph1999continuous,grosshans2002continuous, weedbrook2004quantum, bian2026approaching} and discrete modulation \cite{leverrier2009unconditional,li2018user,DMCVLeverrier, DMCVLinjie}.
The Gaussian-modulated protocol is theoretically optimal in terms of secret key rate (SKR).
In recent years, the Gaussian-modulated protocol has demonstrated strong performance in metropolitan-area links \cite{zhang2020long, jain2022practical, hajomer2024long, wang2025high}, free-space links \cite{FreeSpace2026long}, and quantum access networks \cite{hajomer2024passive, pan2025high}. At the same time, with rapid progress in digital signal processing \cite{fan2026practical} and photonic integration \cite{Bian2024Continuous, pietri2024experimental}, high-speed and low-cost CV-QKD implementations are now within reach.

Nevertheless, perfect Gaussian modulation is practically unattainable \cite{jouguet2012analysis}, especially in high-speed systems. As modulation baud rates increase, the finite resolution of hardware becomes a fundamental limitation, making discrete modulation an inevitable alternative for robust, high-speed CV-QKD systems \cite{hajomer2024continuous10GBaud, ng2026gigabit, wu2026high}. 
However, the non-Gaussian nature of the shared quantum state complicates security analysis, which is conventionally based on covariance matrix and the optimality of Gaussian states \cite{garcia2006unconditional, wolf2006extremality}. 
Current security analysis methods for discrete-modulated (DM) CV-QKD use semidefinite programming (SDP) to derive numerical bounds under multiple constraints \cite{DMCVLeverrier, DMCVLinjie, denys2021explicit}. 
While mathematically rigorous, such approaches often incur prohibitive computational overhead as the constellation size scales up, and obscures the direct physical insights into non-Gaussian modulation.

In this Letter, we propose a practical security analysis framework for arbitrary DM CV-QKD relying solely on the uncertainty principle constraint. 
By introducing a heterodyne detected ancillary mode in the entangled source, we increase the quantifiable correlation between the two legitimate parties, thereby tightening the estimation on the accessible information to a potential eavesdropper. This source-mode expansion allows us to utilize the most basic constraint to provide a feasible set of covariance matrix characterizing the overall system, in which the `worst' matrix determines the SKR lower bound. 
Furthermore, by exploiting the convexity of SKR, we reduce the worst-case search process by establishing a tightened, analytical search boundary. It enables both numerical and analytical calculation under high-order constellations.
We show that the performance of the proposed protocol approaches that of ideal Gaussian modulation with no more than 256 constellation points. Experimental demonstrations on both discrete-component and chip-based platforms achieve Mbps-level SKRs. Moreover, our approach also provides an interpretable correspondence between the modulation constellation geometry and the protocol performance.

\begin{figure}
    \centering
    \includegraphics[width= 7.5 cm]{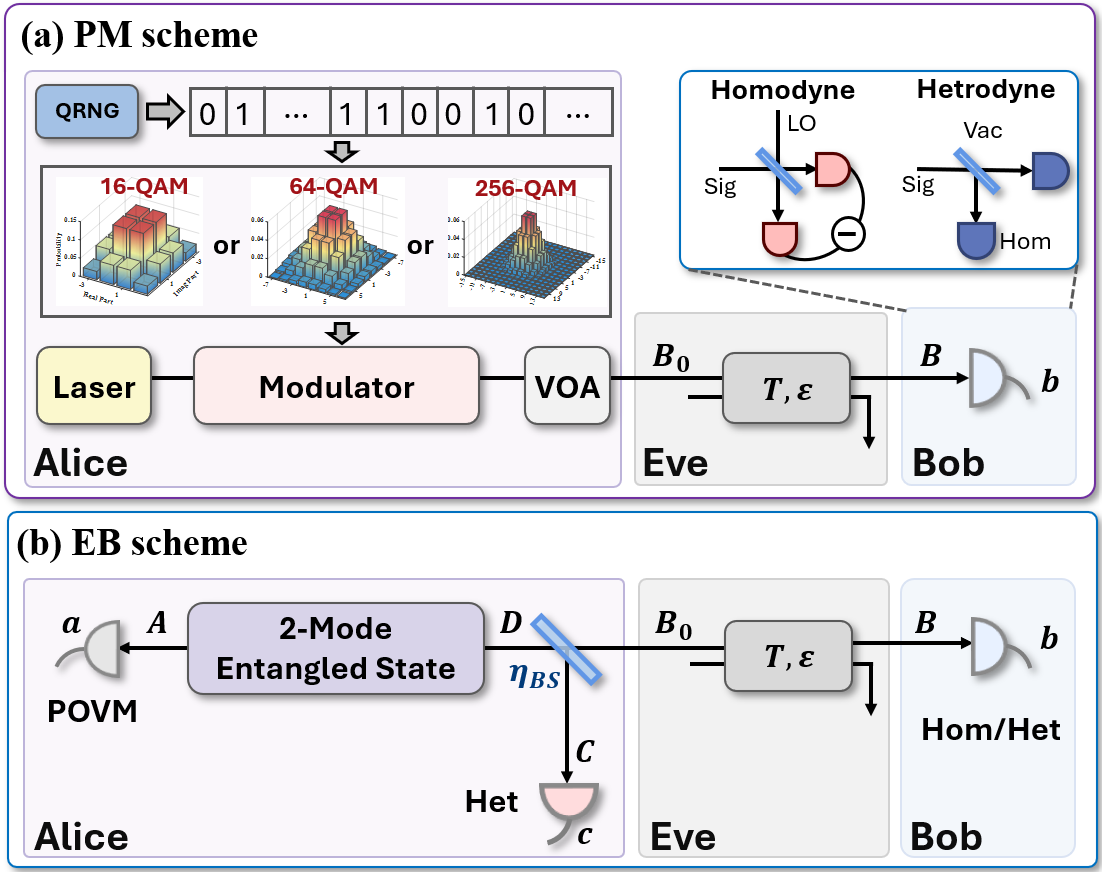}
    \caption{(a) The prepare-and-measure scheme of the proposed DM CV-QKD protocol. A quantum random number generator (QRNG) produces \(n\)-bit symbols, each of which is mapped to one of $2^n$ coherent states in the discrete modulation constellation. Inset plots show the probability distributions of three representative modulation formats: 16-QAM, 64-QAM, and 256-QAM. VOA, variable optical attenuator; Hom, homodyne measurement.
    (b) The entanglement-based scheme for the proposed protocol. 
    POVM, positive-operator valued measure. Het, heterodyne measurement.
    }\label{ProtocolEB}
\end{figure}

\textit{Protocol and security analysis---}
The prepare-and-measure (PM) scheme of the proposed protocol is illustrated in Fig. \ref{ProtocolEB} (a).
Alice prepares a coherent state $\rho_{B_0}$, randomly chosen from an $n$-coherent-state set \(\{\alpha_{1},\alpha_{2},\dots,\alpha_{n}\}\) with discrete probability distribution \(\{p_1,p_2,\dots,p_n\}\). Alice sends the state to Bob, and records the mean of quadratures for \(\rho_{B_0}\), as $(x_{B_0}, p_{B_0})$. Bob receives the state $\rho_{B}$ and performs homodyne or heterodyne measurement, obtaining the results $b$. Alice and Bob perform parameter estimation, information reconciliation, and privacy amplification to extract secret keys.




Building equivalent entanglement-based (EB) scheme is necessary for security analysis. While if we follow the traditional two-mode entangled source $(\rho_{AB_0})$ modeling method, Alice should perform a positive-operator valued measure (POVM) with finite elements rather than heterodyne detection on mode $A$. As a result, the covariance matrix characterizing the correlation between Alice and Bob cannot be directly constructed from the measured data, which complicates the security analysis.

We follow the intuition of QKD security originating from the monogamy of entanglement, namely that the stronger the quantifiable correlation between Alice and Bob, the less information is leaked to Eve. Specifically, we introduce an ancillary mode \(C\) on Alice's side to construct a multi-mode entangled source that enhances the quantifiable correlations with Bob. 
Mode $C$ is heterodyne measured, enabling to estimate the covariance matrix between $C$ and $B$. Meanwhile, modes $A$ and $C$ are internally correlated, thus helping establish the covariance matrix between Alice and Bob, $\gamma_{ACB}$, from which the maximum Holevo information $S_{BE}$ between Bob and Eve can be bounded.
$\gamma_{ACB}$ is expressed with several sub-matrices,
\begin{equation} 
    \gamma_{ACB}=
    \begin{pmatrix} 
            \gamma_A & \psi _{AC} & \kappa_{AB} \\ 
    \psi _{AC}^T & \gamma_C & \psi _{CB}\\
    \kappa_{AB}^T & \psi _{CB}^T &\gamma_B
    \end{pmatrix},
\end{equation} 
where $\gamma_A$, $\gamma_C$, and $\psi _{AC}$ can be theoretically calculated based on the entangled source. $\psi _{CB}$ and $\gamma_B$ can be estimated through the measured data. Only $\kappa_{AB} =(\kappa_{11}\ \kappa_{12}\ ;\ \kappa_{21}\ \kappa_{22})$ is unknown due to the POVM measurement, rendering $S_{BE}$ a function of $\kappa_{AB}$. Nevertheless, the uncertainty principle \cite{weedbrook2012gaussian} imposes a constraint on $\gamma_{ACB}$, namely $\gamma_{ACB}+i\Omega \ge 0$,
which restricts the value of $\kappa_{AB}$ within a set $\mathcal{F}_{\kappa_{AB}}$. Therefore, we can find the maximum information leakage $S_{BE}(\kappa_{AB}^{\mathbf{worst}})$ by traversing $\mathcal{F}_{\kappa_{AB}}$, and the SKR is defined as
\begin{equation} \label{Eq_SKR}
    K=\beta I(a:b)-\sup_{\kappa_{AB} \in \mathcal{F}_{\kappa_{AB}}}S_{BE}(\kappa_{AB}).
\end{equation} 

We develop a simple and effective EB scheme, in which the state sent to Bob is only decided by Alice's POVM result $a$. This means the sub-state $\rho_{CB_0}^a$ conditioned on $a$ is a product state, $\rho_{CB_0}^a=\rho_{C}^a\otimes\rho_{B_0}^a$ \cite{Supplemental}. So the mean values of $\rho_C^a$ can be designed to linearly depend on $\rho_{B_0}^a$. Following this way, one only needs to find a two-mode entangled state $|\Psi\rangle_{AD}$ and then let mode $D$ pass through a beam-splitter with transmittance $\eta_{BS}$, yielding modes $C$ and $B_0$, as shown in Fig. \ref{ProtocolEB} (b). The source is designed as $|\Psi\rangle_{AD}=\sum_{i=1}^{n}\sqrt{\nu_i}|\varphi_A^i\rangle |\theta_D^i\rangle $, where $|\varphi_A^i\rangle$ and $|\theta_D^i\rangle$ are orthogonal states that diagonalize the mixed states $\rho_A$ and $\rho_D$. The measurements on modes $A$ and $C$ project mode $B_0$ onto a state $\rho_{B_0}^{a,c}$, which corresponds to $\rho_{B_0}$ prepared in the PM scheme. The sending state is therefore specified by the source $|\Psi\rangle_{AD}$ and $\eta_{BS}$.

Moreover, we can find a simplified form for the heterodyne result of mode $C$ to estimate $\psi_{CB}$. Originally, the heterodyne result consists of the mean values of the quadratures for \(\rho_C\) and a random fluctuation, requiring another set of random numbers in the PM model.
However, since $\rho_{CB_0}^a$ is a product state, only the mean values of \(\rho_C\) are required \cite{Supplemental}, and the mean values can be directly calculated with the sending state $\rho_{B_0}$. Therefore, the PM scheme has no change in hardware comparing to the existing one-way CV-QKD system. 

We next introduce the numerical and analytical methods for SKR calculation. Since the direct evaluation of SKR with Eq. \ref{Eq_SKR} requires an exhaustive search over four unknowns of $\kappa_{AB}$, it suffers from a significant computational burden. To simplify the searching process, we symmetrize $\gamma_{ACB}$ into the symmetric form \(\gamma_{ACB}^{\mathrm{sym}}\), which includes only one unknown parameter $\kappa$, as shown in Fig. \ref{CMSym} (a). This approach not only simplifies the search process but also enables the derivation of an analytical boundary. We subsequently present both numerical and analytical methods for key rate calculation based on the analytical boundary.
The key step is the symmetrization operation, which is outlined in the following proposition.

\begin{figure}
    \centering
    \includegraphics[width= 8 cm]{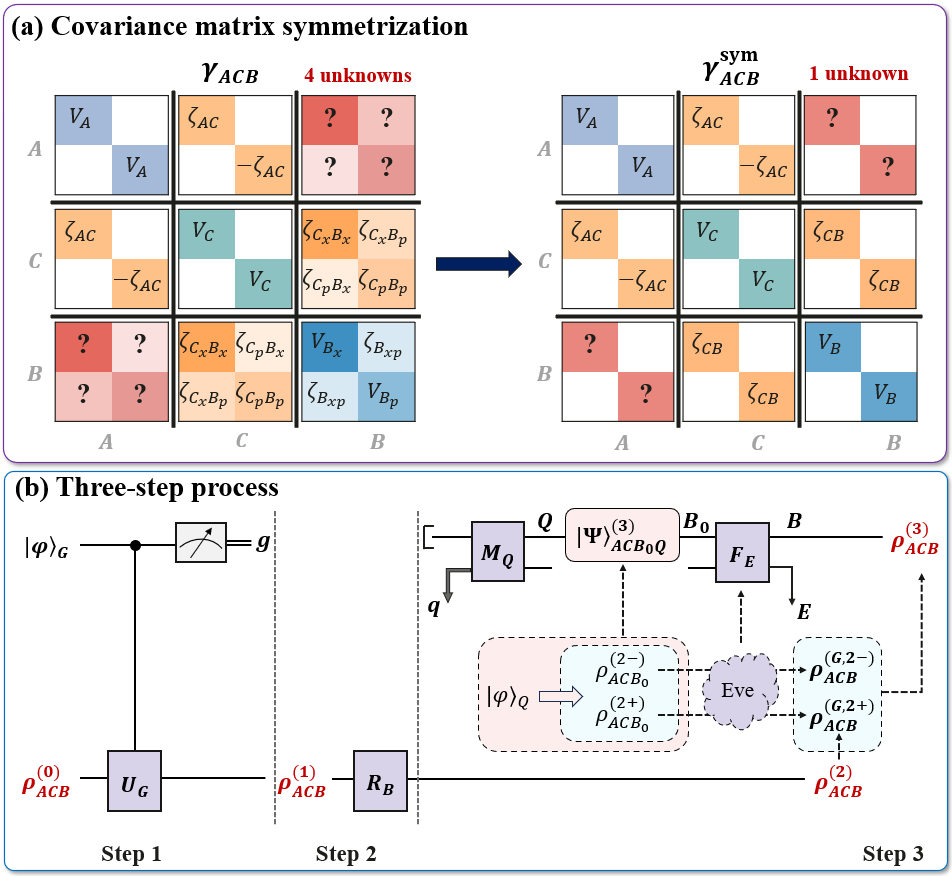}
    \caption{(a) Covariance matrix and its symmetric form after symmetrization. The white boxes represents 0. (b) Three steps for covariance matrix symmetrization. $\mathit{S}_1$: Controlled rotation. $\mathit{S}_2$: Local rotation. $\mathit{S}_3$: State mixture.
    }\label{CMSym}
\end{figure}


\begin{proposition}
Suppose the initial three-mode quantum state is $\rho_{ACB}^{(0)}$ with covariance matrix 
    \begin{equation} 
    \gamma_{ACB}^{(0)}=\gamma_{ACB}=
    \begin{pmatrix} 
    V_{A}I_{2} & \phi_{AC}\,\sigma_{Z} & \kappa_{AB}^{(0)} \\ 
    \phi_{AC}\,\sigma_{Z} & V_CI_{2} & \psi _{CB}^{(0)}\\
    (\kappa_{AB}^{(0)})^T & (\psi _{CB}^{(0)})^T &\gamma_B^{(0)}
    \end{pmatrix},
\end{equation} 
where the upper-left $4\times 4$ sub-matrix $\gamma_{AC}$ is already symmetric.
The remaining sub-matrices are $\kappa_{AB}^{(0)}=(\kappa_{11}\ \kappa_{12}\ ;\ \kappa_{21}\ \kappa_{22})$, $\psi _{CB}^{(0)}=(\phi_{11}\ \phi_{12}\ ;\ \phi_{21}\ \phi_{22})$, $\gamma_B^{(0)}=(b_{11}\ b_{12}\ ;\ b_{12}\ b_{22})$. Then there exist a state $\rho_{ACB}^{\mathrm{sym}}$, whose covariance matrix takes the symmetric form
\begin{equation}
\gamma_{ACB}^{\mathrm{sym}}=
\begin{pmatrix}
V_A I_2 & \phi_{AC}\sigma_Z & \kappa \sigma_Z \\
\phi_{AC}\sigma_Z & V_C I_2 & \phi I_2 \\
\kappa \sigma_Z & \phi I_2 & V_B I_2
\end{pmatrix},
\end{equation}
where $\phi=\sqrt{(\frac{\phi_{11}+\phi_{22}}{2})^2+(\frac{\phi_{12}-\phi_{21}}{2})^2}$, $\kappa=\frac{(\kappa_{11}-\kappa_{22})(\phi_{11}+\phi_{22})}{4\phi}
-\frac{(\kappa_{12}+\kappa_{21})(\phi_{12}-\phi_{21})}{4\phi}$, $V_B=(b_{11}+b_{22})/2$. 
The corresponding SKRs satisfy $K\big(\rho_{ACB}^{\mathrm{sym}}\big) \le K\big(\rho^{(0)}_{ACB}\big)$.
\end{proposition}

\begin{proof}
The proof relies on an important lemma, which is the convexity of the SKR: for a mixture of states with same SKRs \cite{Supplemental},
\begin{equation} 
    K\big(\sum_{i=1}^2 p_i \rho_i\big) \le \sum_{i=1}^2 p_iK(\rho_i)=K(\rho_i),
\end{equation} 
where $\sum_i p_i=1$. 
The symmetrization consists of three steps \cite{Supplemental}, as illustrated in Fig.~\ref{CMSym}(b).

\textit{$\mathbf{Step}\ 1$: Controlled rotation.}
Apply a controlled rotation $U_G=|0\rangle\langle0|_G\otimes I_{ABC} +|1\rangle\langle1|_G\otimes R_A\!\left(\tfrac{\pi}{2}\right)\otimes R_{BC}\!\left(-\tfrac{\pi}{2}\right)$, yielding the equally mixed state and covariance matrix
\begin{equation}
   \rho_{ACB}^{(1)}=\frac{\rho_{ACB}^{(0)}+\rho_{ACB}^{\mathrm{(0R)}}}{2},\quad \gamma_{ACB}^{(1)}=\frac{\gamma_{ACB}^{(0)}+\gamma_{ACB}^{\mathrm{(0R)}}}{2},
\end{equation} 
where the new sub-matrices are $\gamma_{B}^{(1)}=V_BI_2$ and $V_B=(b_{11}+b_{22})/2$, $\kappa_{AB}^{(1)}=(\kappa^{(1)}\ \Delta\kappa^{(1)} \ ;\ \Delta\kappa^{(1)}\ -\kappa^{(1)})$, $\psi _{CB}^{(1)}=(\phi^{(1)}\ \Delta\phi^{(1)} \ ;\ -\Delta\phi^{(1)}\ \phi^{(1)})$, where $\kappa^{(1)}=(\kappa_{11}-\kappa_{22})/2$, $\Delta\kappa^{(1)}=(\kappa_{12}+\kappa_{21})/2$, $\phi^{(1)}=(\phi_{11}+\phi_{22})/2$, $\Delta\phi^{(1)}=(\phi_{12}-\phi_{21})/2$. It's obvious that $K(\rho_{ACB}^{(0)})=K(\rho_{ACB}^{\mathrm{(0R)}})$ because the rotation is a unitary operation. With the convexity of SKR, we have $K\big(\rho_{ACB}^{(1)}\big) \le K\big(\rho^{(0)}_{ACB}\big)$.

\textit{$\mathbf{Step}\ 2$: Local rotation.} 
A local rotation $R_B$ with symplectic matrix $X=(\cos \theta\ -\sin \theta\ ;\ \sin \theta\ \cos \theta)$ is applied to mode $B$, yielding $\gamma_{ACB}^{(2)}=(I_{4} \oplus X)\gamma_{ACB}^{(1)}(I_{4} \oplus X)^T$.
Choosing $\tan \theta= (\Delta\phi^{(1)})/(\phi^{(1)})$, the sub-matrix is symmetrized as $\psi _{CB}^{(2)}=X\psi _{CB}^{(1)}X^T=\phi I_2$, and $\phi =\sqrt{(\phi^{(1)})^2 + (\Delta\phi^{(1)})^2}$. 
Similarly, $\kappa _{AB}^{(2)}=X\kappa _{AB}^{(1)}X^T=(\kappa^{(2)}\ \Delta\kappa^{(2)}\ ;\ \Delta\kappa^{(2)}\ -\kappa^{(2)})$,
where $\kappa^{(2)}=\kappa^{(1)} \cos \theta + \Delta\kappa^{(1)} \sin \theta$, and $\Delta\kappa^{(2)}=-\kappa^{(1)} \sin \theta + \Delta\kappa^{(1)} \cos \theta$.
This unitary operation does not change SKR, therefore we have $K\big(\rho_{ACB}^{(2)}\big) = K\big(\rho_{ACB}^{(1)}\big)$.

\textit{$\mathbf{Step}\ 3$: State mixture.}
Consider a Gaussian state $\rho_{ACB}^{(G,2+)}$ having covariance matrix $\gamma_{ACB}^{(2+)}$ ($=\gamma_{ACB}^{(2)}$). Let $\rho_{ACB}^{(G,2-)}$ denote a Gaussian state with covariance matrix $\gamma_{ACB}^{(2-)}$, obtained by replacing $\Delta\kappa^{(2)}$ in $\gamma_{ACB}^{(2+)}$ with $-\Delta\kappa^{(2)}$. The equally mixed state and its covariance matrix are
\begin{equation}
    \rho_{ACB}^{\mathrm{sym}}=\frac{\rho_{ACB}^{(G,2+)}+\rho_{ACB}^{(G,2-)}}{2},\quad  \gamma_{ACB}^{\mathrm{sym}}=\frac{\gamma_{ACB}^{(2+)}+\gamma_{ACB}^{\mathrm{(2-)}}}{2},
\end{equation}
where the sub-matrix $\kappa _{AB}^{\mathrm{sym}}=\kappa \sigma_Z$ and $\kappa=\kappa^{(2)}$. Since the symplectic eigenvalues for $\gamma^{(2+)}_{ACB}$ and $\gamma_{ACB}^{(2-)}$ are the same, then $\rho_{ACB}^{(G,2+)}$ and $\rho_{ACB}^{(G,2-)}$ yield the same SKR. Following the convexity of SKR and the Gaussian optimality theorem, $K\big(\rho_{ACB}^{\mathrm{sym}}\big) \le K\big(\rho_{ACB}^{(G,2+)}\big) \le K\big(\rho_{ACB}^{(2)}\big)$.

Ultimately, we obtain a scaling chain of SKR 
\begin{equation}
        K\big(\rho_{ACB}^{\mathrm{sym}}\big) \le K\big(\rho_{ACB}^{(2)}\big)
        =K\big(\rho_{ACB}^{(1)}\big) \le K\big(\rho^{(0)}_{ACB}\big),
\end{equation}
which demonstrates that we can calculate the SKR lower bound with a symmetric matrix, thus reducing the problem to solving for the single parameter $\kappa$. 
\end{proof}

Regarding the analytical boundary, by solving the uncertainty-principle constraint, $\gamma_{ACB}^{\mathrm{sym}}+i\Omega \ge 0$, the interval of $\kappa$ is obtained as $\kappa \in \mathcal{F}_k=[\bar{\kappa}-R, \bar{\kappa}+R]$, where 
\begin{subequations}
    \begin{align}
        \bar{\kappa} &= \sqrt{T_{cb}\,\eta_{A}\,(1-\eta_{BS})\big(V_{A}^{2}-1\big)},\\
R &= \sqrt{T_{cb}\,\varepsilon_{cb}\,(1-\eta_{A})\,(V_{A}+1)}.
    \end{align}
\end{subequations}
The parameters \(T_{cb}\) and \(\varepsilon_{cb}\) are defined from the symmetrized two-mode covariance matrix \(\gamma_{CB}^{\mathrm{sym}}\) of modes \(B\) and \(C\) as $T_{cb}=\frac{\phi_{CB}^2}{(V_C-1)^2}$ and $\varepsilon_{cb}=\frac{V_B-1}{T_{cb}}-(V_C-1)$ \cite{Supplemental}.
Finally, Eq. \ref{Eq_SKR} can be updated as
\begin{equation} 
    K=\beta I(a:b)-\sup_{\kappa \in [\bar{\kappa}-R, \bar{\kappa}+R]}S_{BE}(\kappa),
\end{equation}
offering a numerical key rate in our scheme.


\begin{figure}[t]
    \centering
    \includegraphics[width= 8.5 cm]{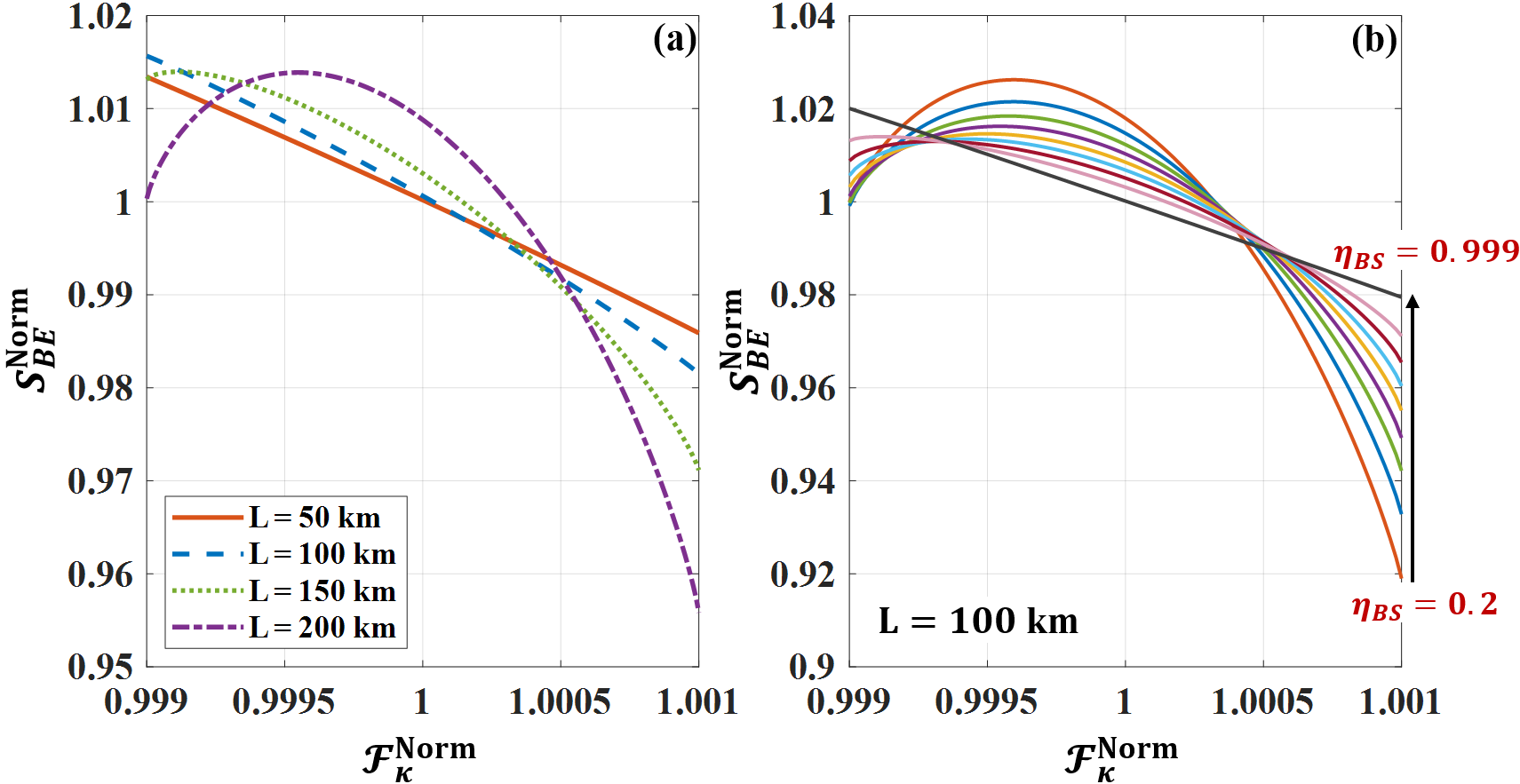}
    \caption{(a) The variation trend of normalized $S_{BE}^{\mathrm{Norm}}$ within the normalized interval $\mathcal{F}_{\kappa}^{\mathrm{Norm}}$ under different transmission distances $L$. (b) The variation trend of normalized $S_{BE}^{\mathrm{Norm}}$ within the normalized interval $\mathcal{F}_{\kappa}^{\mathrm{Norm}}$ under different $\eta_{BS}$ at 100 km. $\mathcal{F}_{\kappa}^{\mathrm{Norm}}=[-R/\bar{\kappa},\ R/\bar{\kappa}]$. 
    Simulation parameters: excess noise $\varepsilon=0.04$ SNU, detection efficiency $\eta=0.6$, electronic noise $\nu_{ele}=0.1$ SNU. 
    The modulation variance $V_M=4$ SNU, and $\eta_A=1-10^{-4}$.
    }\label{SBE}
\end{figure}

Numerically, with a bisection search over the interval $\mathcal{F}_{\kappa}$ to locate the true $\kappa_t$ and the maximum \(S_{BE}^{G}(\kappa_t)\), it requires roughly 16 iterations to reach the desired error tolerance. The tightness of the boundary reduces the computational costs. 

For the analytical lower bound of SKR, one may intuitively expect it to occur at the left boundary $\kappa_l=\bar{\kappa}-R$ of \(\mathcal{F}_{\kappa}\), since weaker correlations between Alice and Bob generally increase Eve's accessible information and thus reduce the SKR. However, this intuition does not always hold for the three-mode covariance matrix. 
As shown in Fig. \ref{SBE} (a), \(S_{BE}^{\mathrm{Norm}}\), which obtained by normalizing \(S_{BE}\) with its mean value, decreases monotonically over the normalized interval $\mathcal{F}_{\kappa}^{\mathrm{Norm}}$ at short distances ($L=50$ and 100 km). While at long distances ($L=150$ and 200 km), it exhibits a nonmonotonic trend, first increasing and then decreasing.
Fig. \ref{SBE} (b) shows the influence of $\eta_{BS}$ on $S_{BE}$ over $\mathcal{F}_{\kappa}^{\mathrm{Norm}}$ at a fixed transmission distance of 100 km, where the modulation variance is set as $V_M=4$ SNU, and $V_A=V_M/\eta_{BS}+1$. 
As $\eta_{BS}$ increases, the variation of $S_{BE}^{\mathrm{Norm}}$ becomes smoother and more monotonic, and the relative deviation between the true worst-case point and the left boundary \(\kappa_{\mathrm{left}}\) roughly decreases. In the limit $\eta_{BS} \to 1$, the worst-case point approaches the left boundary, indicating that the characteristics of the three-mode matrix gradually converge to those of a two-mode matrix.
To gain insight into this phenomenon, we apply an entanglement attack model \cite{Supplemental}. When approaching the left physical boundary, Eve's state becomes nontrivial and entangled, which may partially weakens the eavesdropper's ability. However, the specific characteristics of the three-mode state still require further exploration.

\begin{figure}[t]
    \centering
    \includegraphics[width= 8.5 cm]{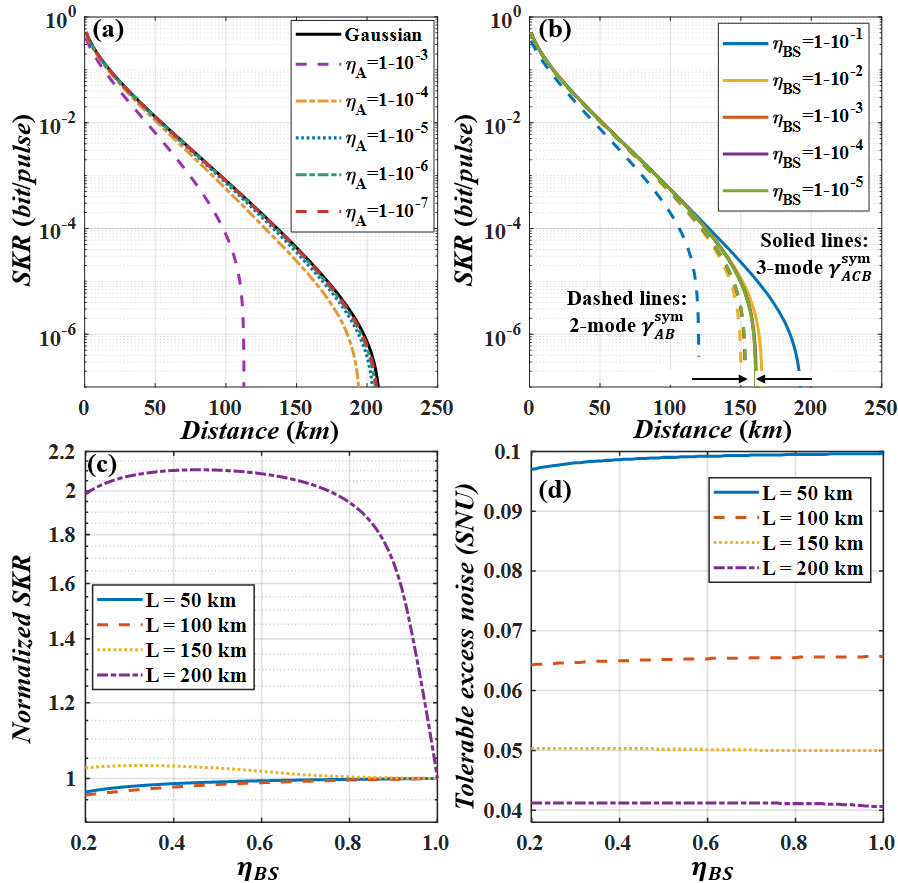}
    \caption{
        (a) SKR of the DM protocol for different $\eta_A$ compared with the ideal Gaussian protocol, with a fixed mode variance $V_A=5$ SNU and $\eta_{BS}=0.9$. (b) SKR based on 3-mode matrix $\gamma_{ACB}^{\mathrm{sym}}$ and 2-mode matrix $\gamma_{AB}^{\mathrm{sym}}$ with different $\eta_{BS}$. (c) Normalized SKR versus \(\eta_{BS}\) under different transmission distances, where the SKR is normalized to its value at \(\eta_{BS}=1-10^{-5}\). (d) Tolerable excess noise versus \(\eta_{BS}\) under different distances.
    Simulation parameters: reconciliation efficiency $\beta=0.95$, excess noise $\varepsilon=0.04$ SNU, detection efficiency $\eta=0.6$, electronic noise $\nu_{ele}=0.1$ SNU \cite{Supplemental}. $V_M=4$ and $\eta_A=1-10^{-4}$ for (b), $V_M=4$ and $\eta_A=1-10^{-6}$ for (c) and (d).
    }\label{SimulationPerformance}
\end{figure}

To obtain a rigorous analytical key rate lower bound, we discard mode \(C\) from the 3-mode covariance matrix $\gamma_{ACB}^{\mathrm{sym}}$ and obtain the 2-mode matrix \(\gamma_{AB}^{\mathrm{sym}}\). Then the lower bound of SKR based on \(\gamma_{AB}^{\mathrm{sym}}\) can be calculated analytically at the left boundary $\kappa_l=\bar{\kappa}-R$. This may come at the cost of a reduced key rate, with less than $2\%$ degradation for distances below 150 km when $\eta_{BS}$ approaches 1, whereas a more noticeable performance loss is observed beyond 150 km.

\begin{figure*}[t]
    \includegraphics[width= 17 cm]{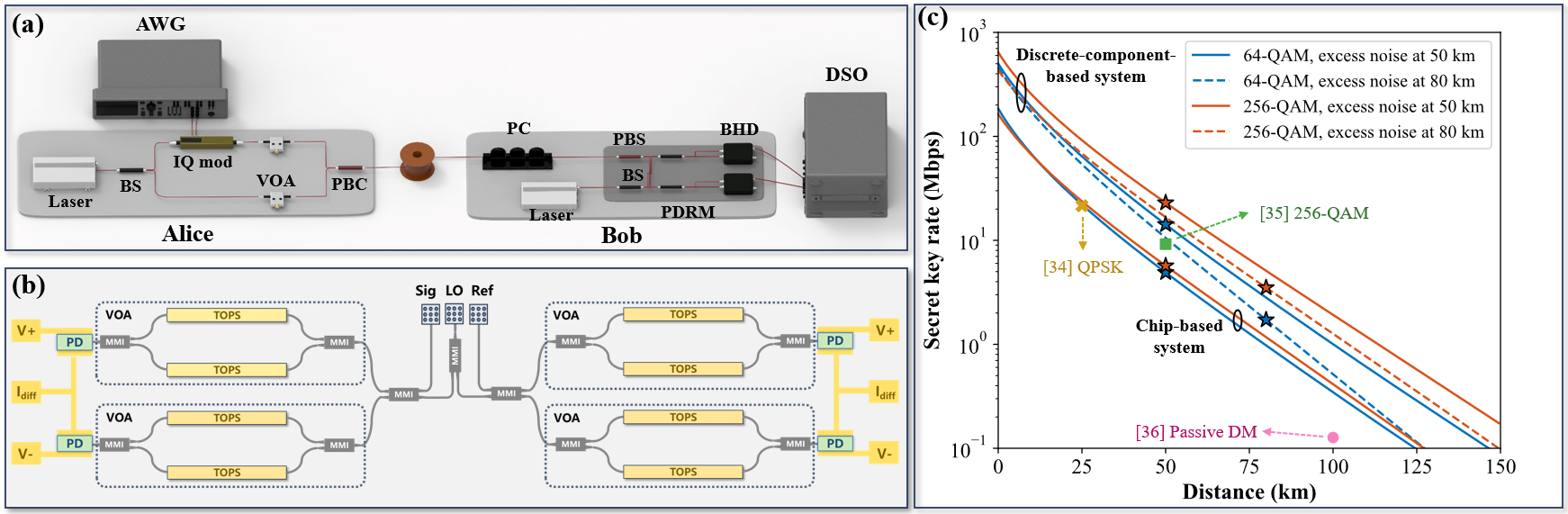}
    \caption{\label{Exp_Scheme}
    (a) Experimental setup of the DM CV-QKD system based on discrete component. AWG, arbitrary waveform generator; IQ mod, in-phase/quadrature modulator; VOA, variable optical attenuator; PC, polarization controller; PBC, polarization beam combiner; PBS, polarization beam splitter; BS, beam splitter; PMOC, polarization-maintaining optical couplers; BHD, balanced homodyne detector; PDRM, polarization diversity receiver module; DSO, digital storage oscilloscope.
    (b) The integrated silicon photonic receiver chip. PD, photodiode; MMI, multi-mode interferometer. 
    (c) Experimental key rates and numerical simulations. The solid and dashed lines are the numerical simulations of SKRs using experimental parameters in discrete-component-based and chip-based system, along with six pointed stars corresponding to the experimental results, which are compared with the existing experimental results in different modulation schemes represented by yellow cross \cite{wang2022sub}, green square \cite{pan2022experimental}, pink circle \cite{PassiveDM2026}.
    }
\end{figure*}

\textit{Results---}
Numerical simulations are conducted to present the performance of the protocol. 
Firstly, A key factor affecting SKR is the design of entangled source. We define a dimensionless parameter $\eta_A$ to quantify the deviation between $|\Psi\rangle_{AD}$ and an ideal two-mode squeezed-vacuum state $|\Psi\rangle_{AA_0}$ with the same $V_A$ \cite{Supplemental}. As shown in Fig. \ref{SimulationPerformance} (a), for a fixed mode variance $V_A=5$, larger $\eta_A$ yields better performance. The results further indicate that $\eta_A$ should be at least on the order of $1-10^{-4}$ to avoid significant performance degradation. 
Moreover, when $\eta_A$ is on the order of $1-10^{-7}$, the SKR becomes nearly indistinguishable from that of the ideal Gaussian protocol. Such a source can be realized using no more than 256 constellation points \cite{Supplemental}.
Fig.~\ref{SimulationPerformance}(b) shows the SKR under different \(\eta_{BS}\) calculated from the 3-mode matrix $\gamma_{ACB}^{\mathrm{sym}}$ and 2-mode matrix $\gamma_{AB}^{\mathrm{sym}}$, where the same transmitted states are considered with a 64-QAM constellation and modulation variance \(V_M=4\). Overall, the SKR from $\gamma_{AB}^{\mathrm{sym}}$ is slightly lower than that from $\gamma_{ACB}^{\mathrm{sym}}$. As \(\eta_{BS}\) increases, the SKR of the two-mode scheme gradually increases, while for the three-mode scheme, the SKR decreases at long distances but increases at short distances. Consequently, the discrepancy between the SKRs gradually diminishes, and the corresponding curves eventually converge to a same one.
Fig.~\ref{SimulationPerformance} (c) and (d) show the normalized SKR and the maximum tolerable excess noise as functions of \(\eta_{BS}\) under different transmission distances. Both quantities exhibit similar behavior: increasing \(\eta_{BS}\) generally improves the performance, particularly in the short-distance regime, while the improvement gradually saturates as \(\eta_{BS}\rightarrow 1\). At longer distances, the enhancement becomes weaker and may even slightly decrease in the high-\(\eta_{BS}\) regime. This reflects the trade-off introduced by the beam-splitter transmittance.

The proposed DM protocol is experimentally demonstrated on both a discrete-component platform and a silicon photonic integrated platform.
Fig.~\ref{Exp_Scheme}(a) shows the experimental setup of the discrete-component system employing a polarization diversity receiver module (PDRM). In the integrated platform, the PDRM is replaced by a silicon photonic receiver chip, as illustrated in Fig.~\ref{Exp_Scheme}(b).
Discrete modulation is implemented using an IQ modulator followed by a variable optical attenuator. After transmission through the fiber link, the quantum signal is detected at the receiver. For simplicity and accuracy, the one-time shot noise unit calibration method is adopted \cite{zhang2020one}. The detection efficiencies of the PDRM and the silicon photonic receiver are 0.6366 and 0.2196.
The experimental results are shown in Fig. \ref{Exp_Scheme} (c). The four upper curves display the simulated SKRs for the discrete-module-based system under experimental parameters, while the two lower curves correspond to the simulated results for the chip-based system. The experiments are conducted over fiber optic distances of 50 km and 80 km, using 64-QAM and 256-QAM modulation formats, with reconciliation efficiencies of 0.92 and 0.95, respectively.
The experimental key rates are marked with stars, and compared with existing literatures \cite{wang2022sub, pan2022experimental, PassiveDM2026}. It can be seen that, within our framework, the 256-QAM scheme demonstrates a significant advantage in performance compared to other existing results, achieving SKRs of 22.164 Mbps at 50 km and 3.642 Mbps at 80 km. The 64-QAM scheme remains advantageous at 50 km, but its performance degrades more rapidly with distance. The chip-based receiver exhibits higher loss, yielding lower SKRs than the discrete-component system, while using high-order constellation schemes still enables effective key generation.

\textit{Discussions---} 
In this work, we have theoretically proposed and experimentally demonstrated a practical security analysis framework for DM CV-QKD based on source-mode expansion and the Heisenberg uncertainty principle. By introducing an ancillary mode at the source, the security analysis is transformed from a complex optimization problem into a direct evaluation over a physically constrained feasible set, enabling fast calculation of SKR from experimentally accessible covariance matrix. The framework is applicable to arbitrary discrete constellations and can achieve high performance approaching that of ideal Gaussian modulation. Beyond providing a practical security bound, it also establishes an intuitive connection between constellation geometry and protocol performance, offering new insights into the design of discrete-modulation formats. Although the present work focuses on QAM constellations, the framework can be readily extended to other modulation formats. In particular, understanding how asymmetric constellations, such as PSK, influence security performance remains an interesting direction for future investigation.

Our results provide a practical route toward ultra-high-speed, low-cost, and integrated CV-QKD systems. Discrete modulation is inherently compatible with high-speed electro-optic hardware and existing optical communication infrastructure, facilitating higher baud rates while reducing implementation complexity and deployment costs. Furthermore, the successful demonstration on a silicon photonic receiver platform highlights the feasibility of chip-scale integrated CV-QKD and identifies key engineering parameters, including coupling loss and detection efficiency, that directly impact transmission distance and secret key rate. Looking forward, the proposed framework can be naturally extended to multi-user quantum access networks and other large-scale quantum communication architectures, thereby supporting the widespread deployment and commercialization of CV-QKD technologies.

This research was supported by the National Cryptologic Science Fund of China (2025NCSF02050), the National Natural Science Foundation of China (U24B20135).

\bibliography{apssamp}

\clearpage

\appendix

\onecolumngrid
\section*{Appendix for ``Discrete-Modulated Continuous-Variable Quantum Key Distribution with Uncertainty Principle"}
\twocolumngrid

\section{Theoretical methods}
This section provides the theoretical details. Under our framework, it is necessary to introduce ancillary entangled modes into the source and perform heterodyne detection on them to assist the security analysis. We first generalize the source to an $N$-mode entangled source $|\Psi\rangle_{A\mathbf{C}B_0}$ in Sec. \ref{Nmode}, where $\mathbf{C}=C_1C_2\dots C_{N-2}\ (N>2)$. 
Then we prove the product feature of the sub-state $\rho_{\mathbf{C}B_0}^a$ conditioned on $a$ in Sec. \ref{POVM}. By exploiting this feature, we provide a simple representation of $\rho_{\mathbf{C}}$ in Sec. \ref{simply_rhoC}, which removes the need of practical simulation of $c$ in the PM scheme. 
Then in Sec. \ref{SourceDesign}, we explain the design for three-mode entangled source and present the optimal results for source parameters under quadrature-amplitude modulation (QAM) cases. 
For the security analysis, we first derive the SKR formula based on the covariance matrix in Sec. \ref{SKR}, then we prove the convexity of SKR in Sec. \ref{Convexity}. With this poverty, we accomplish the symmetrization of the 3-mode covariance matrix in Sec. \ref{symmetrization}, and solve the only unknown parameter $\kappa$ in Sec. \ref{kappa}. Furthermore, to explain the nonlinear variation of the SKR within the range of $\kappa$, we introduce an attack model in Sec. \ref{model} and provide a physical interpretation. We then extend the symmetrization procedure to incorporate detector modeling and demonstrate its compatibility with practical detection schemes in Sec. \ref{Compatibility}. Finally, in Sec. \ref{Simulation} we provide a series of simulations, showing that the schemes using high-order constellations can achieve performance close to that of the ideal Gaussian protocol without increasing computational complexity. These results demonstrate the advantages in both high performance and efficiency, suggesting the potential for large-scale applications of DM CV-QKD systems.

\subsection{$N$-mode entangled source}\label{Nmode}

Consider the EB scheme of a protocol, the state sent by the transmitter Alice usually depends on the entangled source she prepares and the measurements she performs. The most common case is Alice performing heterodyne detection on one mode of an ideal two-mode squeezed vacuum state, thereby projecting the other mode onto a infinite Gaussian state.
While to send a finite discrete-distributed coherent state, POVM measurement other than heterodyne detection should be introduced at the source. However, POVM measurement alone is insufficient to complete the security analysis. So we choose to inject more entangled modes $\mathbf{C}$ at the source to compete with eavesdropper and assist the security analysis. 
The $N$-mode entangled source is shown in Fig. \ref{NmodeEB}. Alice generates an $N$-mode entangled state $|\Psi\rangle_{A\mathbf{C}B_0}$, which is an purification of the mixed state $\rho_{B_0}$, where \(\mathbf{C}=C_1C_2\dots C_{N-2}\ (N > 2)\). Alice keeps modes $A$ and $\mathbf{C}$, and sends mode $B_0$ to the receiver Bob. The measurements for mode $A$ is POVM, obtaining the result $a$, and the measurements for modes $C$ are heterodyne measurements, with results recorded as $\mathbf{c}=\{c_1,\ c_2,\ \dots,\ c_{N-2}\}$. These measurements will project mode $B_0$ onto a state $\rho_{B_0}^{a,\mathbf{c}}=\{\rho_{B_0}^1,\ \rho_{B_0}^2,\ \dots,\ \rho_{B_0}^n,\ \}$.
\begin{figure}
    \centering
    \includegraphics[width= 6 cm]{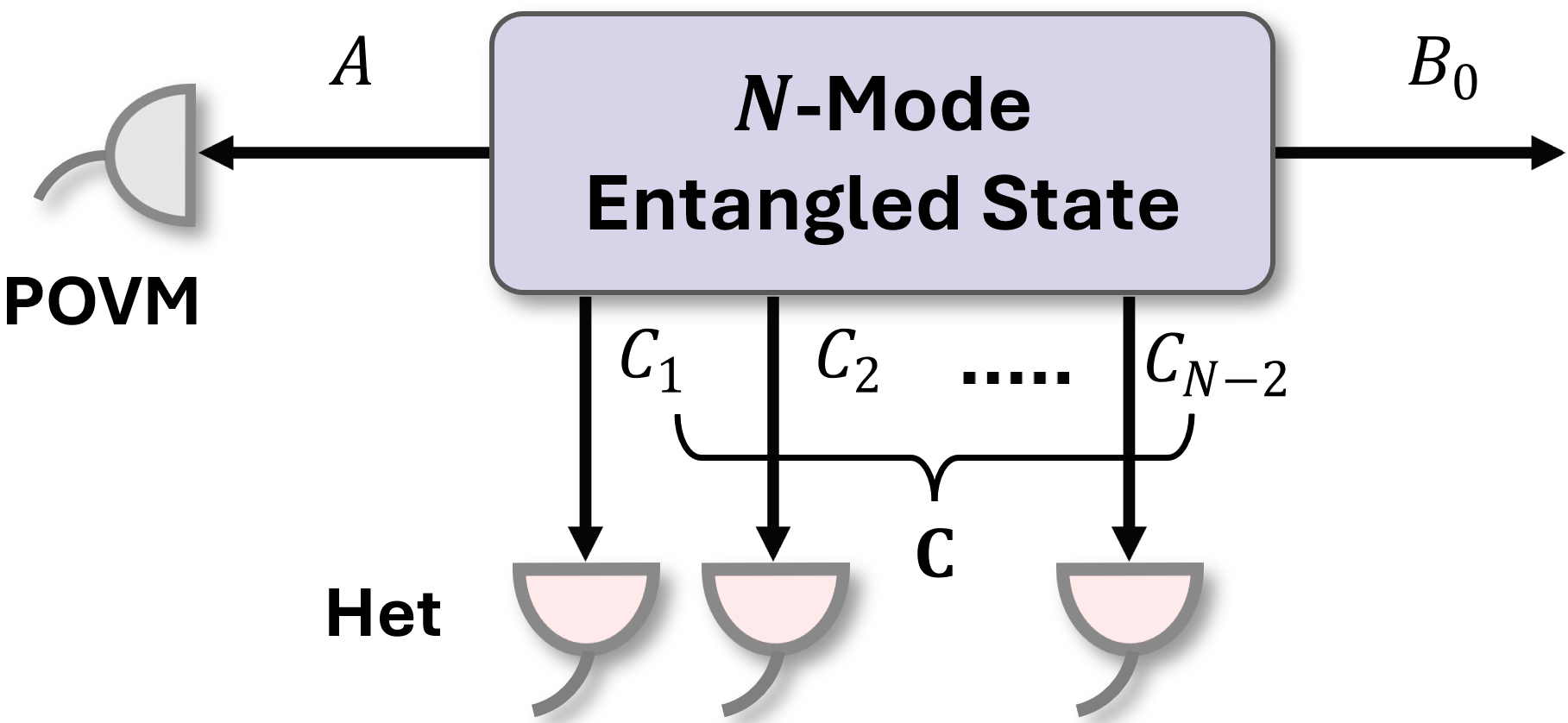}
    \caption{The $N$-mode entangled source. The source generates an \(N\)-mode entangled state, where mode \(A\) is measured by positive-operator valued measure (POVM), and modes \(\mathbf{C}=C_1C_2\dots C_{N-2}\ (N > 2)\) are measured by heterodyne detection (Het). mode \(B_0\) is sent to the receiver side.
    }\label{NmodeEB}
\end{figure}
\subsection{The sub-state $\rho_{\mathbf{C}B_0}^a$ is a product state}\label{POVM}
Since the sending state $\rho_{B_0}$ is only determined by the POVM result $a$ and is independent of the heterodyne detection results $\mathbf{c}$, it is necessary to require that the the sub-state $\rho_{\mathbf{C}B_0}^a$ conditioned on $a$ should be a product state, that is $\rho_{\mathbf{C}B_0}^a=\rho_{\mathbf{C}}^a\otimes\rho_{B_0}^a$.
To explain this necessity, we first prove a lemma which is

\begin{lemma}
For any two-mode entangled state $\rho_{AB}$, if perform heterodyne detection on mode $A$, mode $B$ is projected onto coherent states, then the number of possibly projected coherent states for mode $B$ is either one or infinite.
\end{lemma}

\begin{proof}
We prove this by contradiction. Suppose for $\rho_{AB}$, the heterodyne detection on $A$ can project $B$ onto a finite set of coherent states
\(\{|\alpha_B^1\rangle,\dots,|\alpha_B^n\rangle\}\) with $\infty>n\ge 2$. For generality assume $\rho_{AB}$ a mixed state and exists a purification $|\Psi_{FAB}\rangle=\sum_{i=1}^n \sqrt{p_i}\,|\varphi_{FA}^i\rangle\otimes|\phi_B^i\rangle$. The heterodyne detection projects $A$ onto a coherent state $|\alpha_A\rangle$. 
We divide the entire phase space of mode $A$ into $n+1$ different non-overlapping sets \(\{\mathcal{R}_1,\dots,\mathcal{R}_n,\mathcal{R}_{n+1}\}\). The first \(n\) sets correspond to the \(n\) different output coherent states of mode \(B\), which means for a heterodyne measurement result \(\alpha_A\), if \(\alpha_A\in\mathcal{R}_i\) then mode \(B\) is projected onto \(|\alpha_B^i\rangle\). We assume that for every point in the first \(n\) sets the probability of obtaining the corresponding heterodyne measurement result is nonzero. The final set \(\mathcal{R}_{n+1}\) contains points that never occur as heterodyne results, meaning \(\operatorname{Tr}_B\bigl[\langle\alpha_A|\rho_{AB}|\alpha_A\rangle\bigr]=0\) whenever \(\alpha_A\in\mathcal{R}_{n+1}\). 
For $k\in\{1,\dots,n\}$ and $\alpha_A\in\mathcal{R}_k$, we can obtain
\begin{equation}\label{eq:decomp}
|\alpha_B^k\rangle\langle\alpha_B^k|
=\sum_{i=1}^n\sum_{j=1}^n C_{ij}\,|\alpha_B^i\rangle\langle\alpha_B^j|,
\end{equation}
where $C_{ij}=\sqrt{p_i p_j}/p(\alpha_A)\cdot \operatorname{Tr}_F\bigl[\langle\alpha_A|\varphi_{FA}^i\rangle\langle\varphi_{FA}^j|\alpha_A\rangle\bigr]$, $p(\alpha_A)>0$ is the probability of obtaining result $|\alpha_A\rangle$.

First, we prove that if $i\neq k$ or $j\neq k$, then $C_{ij}=0$, and the only non-zero term is $C_{kk}=1$. The intuitive understanding of this is that there are infinite equations constraining finite variables $C_{ij}$. Apply the displacement operator $D(-\alpha_B^k)$ to $|\alpha_B^k\rangle\langle\alpha_B^k|$, and isolate the vacuum term,
\begin{equation}\label{eq:13}
(1-C_{kk})|0\rangle\langle 0|
= \sum_{\substack{ij=11\\ ij\neq kk}}^{nn} C_{ij}\,|\alpha_B^i-\alpha_B^k\rangle\langle\alpha_B^j-\alpha_B^k|.
\end{equation}
Then calculate the inner product between the Fock state $|t\rangle$ and the vacuum state in Eq. \eqref{eq:13}, we have
\begin{equation}\label{eq:14}
0=\sum_{\substack{ij=11\\ ij\neq kk}}^{nn} C_{ij}\,
e^{-|\beta_i|^2\big/2}\,e^{-|\beta_j|^2\big/2}\,\frac{\beta_i^t\beta_j^t}{t!},\quad \forall t\ge 1,
\end{equation}
where $\beta_i=\alpha_B^i-\alpha_B^k$ and $ \beta_j=\alpha_B^j-\alpha_B^k$.
If we set $\lambda_{ij}=\beta_i\beta_j^*$ and $d_{ij}=C_{ij}\exp\bigl[-(|\beta_i|^2+|\beta_j|^2)/2\bigr]$, then the first $n^2-1$ equations can be written in the matrix formula $\Lambda \mathbf{D}= 0$, where $\lambda$ is the transport of a Vandermonde-like matrix with $n^2-1$ different non-zero $\lambda_{ij}$, and $\mathbf{D}$ is a vector with $n^2-1$ different $d_{ij}$. Since the Vandermonde matrix has the feature of full rank, so the only solution of \(\Lambda\mathbf{D}=0\) is \(d_{ij}=0\), which implies $\forall ij\neq kk,\ C_{ij}=0$, and $C_{kk}=1$.

Second, for any $\mathcal{R}_k$ ($1\le k\le n+1$), we can define a corresponding set $\mathcal{T}_k$ of coherent states, which is $\mathcal{T}_k=\{|\alpha_A\rangle\;|\;\alpha_A\in \mathcal{R}_k\}$. Then the above discussion leads to the conclusion that, for each $\mathcal{T}_k$, there exists at least one state orthogonal to it. First, at the case $1\le k\le n$, in which $C_{ii}=0$ for any $i=j\neq k$. This means $\forall |\alpha_A\rangle\in\mathcal{T}_k$, $\langle\alpha_A|\rho_A^{ii}|\alpha_A\rangle=0$, where $\rho_A^{ii}=\operatorname{Tr}_F\bigl[|\varphi_{FA}^i\rangle\langle\varphi_{FA}^i|\bigr]$. Then $\rho_A^{ii}$ is orthogonal to the set $\mathcal{T}_k$. Second, for the case $k=n+1$, from its definition we know $\operatorname{Tr}_B\bigl[\langle\alpha_A|\rho_{AB}|\alpha_A\rangle\bigr]=0$ if $|\alpha_A\rangle\in \mathcal{R}_{n+1}$. Then $\rho_A=\operatorname{Tr}_B[\rho_{AB}]$ is orthogonal to $\mathcal{T}_{n+1}$.

However, it can be proved that when $n$ is finite, among all these $n+1$ sets $\{\mathcal{T}_1,\mathcal{T}_2,\dots,\mathcal{T}_{n+1}\}$, at least one of them is a complete or over-complete set of the Fock state space, which means no state can be orthogonal to this set. This is contradictory to the previous conclusion. Therefore, for the case $\infty>n\ge 2$, no such a two-mode entangled state can be found. 
\end{proof}

This lemma can be easily generalized to the $N$-mode entangled state case: for any $N$-mode entangled state $\rho_{\mathbf{A}B}$, if after the heterodyne detections for each mode $\mathbf{A}$, mode $B$ is projected onto a coherent state, then the number of possibly projected coherent states for mode $B$ is either one or infinite.

For our EB scheme, the sub-state $\rho_{\mathbf{C}B_0}^a$ after the POVM measurement on mode $A$ faces the same situation as the above argument. Since which coherent state will be sent to Bob is decided by the results of POVM, then the number of possibly projected coherent states for mode $B_0$ after the heterodyne detection on mode $\mathbb{C}$ is only one. This means the conditioned sub-state $\rho_{CB_0}^{a}$ will be a product state, that is $\rho_{CB_0}^{a}=\rho_{C}^{a}\otimes |\alpha_{B_0}^{a}\rangle\langle\alpha_{B_0}^{a}|$.  

The above conclusion illuminates the  logical structure of our framework: First, POVM measurement should be introduced to prepare finite coherent states. Consequently, entangled state with more than
two modes are necessary to assist security analysis. Furthermore, heterodyne detection performed on those auxiliary modes does not change the sending coherent state, which is determined only by the POVM results. Therefore, , the conditioned sub-state after the POVM measurement should be a product state.


\subsection{The simplified representation of $\rho_C$}\label{simply_rhoC}

The secret key rate calculation is based on the covariance matrix $\gamma_{ACB}$, in which $\gamma_A$, $\gamma_C$, and $\psi_{AC}$ can be theoretically calculated, and $\gamma_B$ is estimated only using Bob's data. Only the estimation of $\psi_{CB}$ will use the measurement results of mode $C$. Generally, the measurement result $c$ should be another set of random numbers in the PM model, which can be simulated through a quantum random number generator (QRNG). However, the product feature of sub-state $\rho_{CB_0}^a$ can help to simplify this.

Let's take the x-quadrature of mode $C$ as an example. The state of modes $C$ and $B_0$ is $\rho_{C B_0} = \sum_{i=1}^n p_i\, \rho_{C}^i \otimes \rho_{B_0}^i$.
Using the product structure of the sub-states, the covariance between $\hat{x}_{C}$ and $\hat{x}_{B_0}$ can be written as
\begin{equation}
\label{eq:covariance_trace}
\begin{aligned}
    &\bigl\langle \Delta \hat{x}_{C}\,\Delta \hat{x}_{B_0}\bigr\rangle \\
&= \sum_{i=1}^n p_i\mathrm{Tr}_{C B_0}\Bigl[ \bigl(\hat{x}_{C}-\bar{x}_{C}\bigr)\bigl(\hat{x}_{B_0}-\bar{x}_{B_0}\bigr)\,\rho_{C}^i\otimes\rho_{B_0}^i \Bigr] \\
&=\sum_{i=1}^n p_i\,
\bigl(\bar{x}_{C}^i-\bar{x}_{C}\bigr)\bigl(\bar{x}_{B_0}^i-\bar{x}_{B_0}\bigr),
\end{aligned}
\end{equation}
where $\bar{x}_{C}^i = \mathrm{Tr}_{C}\bigl(\hat{x}_{C}\,\rho_{C}^i\bigr)$ and $\bar{x}_{C} = \sum_{i=1}^n p_i\,\bar{x}_{C}^i$ can be theoretically calculated. Therefore, if $\rho_{B_0}$ has been simulated by a sequence of random number, then it's enough to let the other sequence be $\{\bar{x}_{C}^i,\ \bar{p}_{C}^i\}$ for the calculation of $\psi_{CB}$. This is much simpler than simulating the heterodyne results $c$ using quantum random numbers.

\begin{figure}
    \centering
    \includegraphics[width= 8.5 cm]{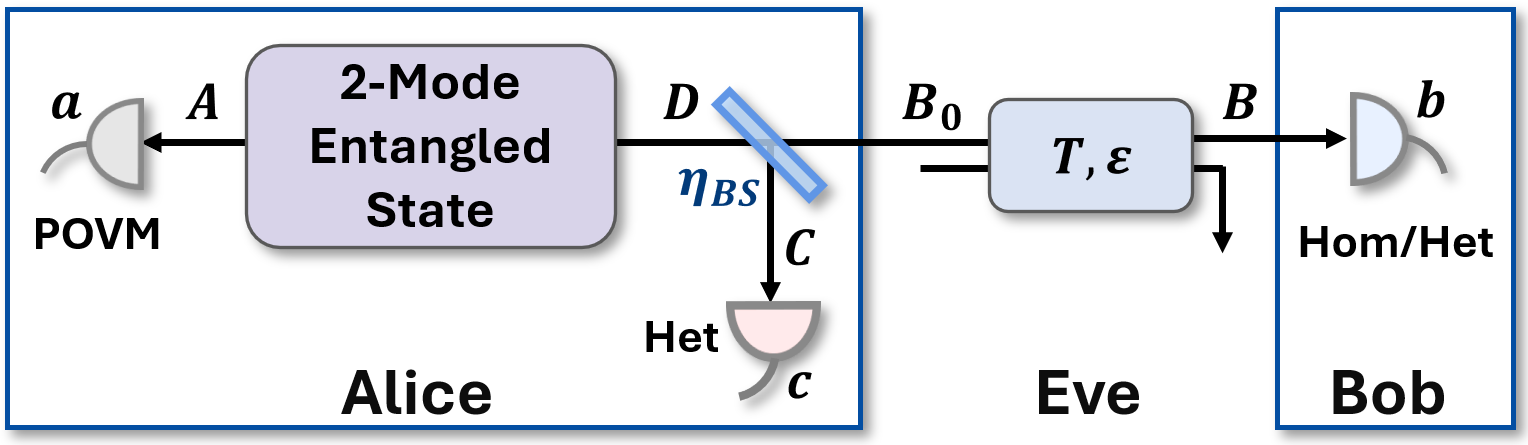}
    \caption{The Entanglement-based (EB) scheme of the proposed DM CV-QKD protocol with a three-mode entangled source. POVM, positive-operator valued measurement; Het, heterodyne measure. Alice prepares a two-mode entangled state $|\Psi_{AD}\rangle$ and performs a POVM measurement on mode $A$, mode $D$ is split by a beam splitter with transmittance $\eta_{BS}$ into modes $C$ and $B_0$. Alice performs heterodyne detection on mode $C$. Mode $B_0$ is transmitted through a quantum channel fully controlled by the eavesdropper Eve with transmittance $T$ and excess noise $\varepsilon$, and heterodyne detected by Bob. Here $a$, $c$ and $b$ are measurement results for $A$, $C$ and $B$ respectively.
    }\label{EB}
\end{figure}

\subsection{Three-mode discrete-modulated source}\label{SourceDesign}

In this subsection, we detail the design principles for an equivalent entanglement source compatible with any discrete-modulation constellation. A three-mode entangled source is adopted as shown in Fig. \ref{EB}, which is not only structurally simple but also highly effective.

\subsubsection{Design of two-mode entangled state}
Our design principle is to let the correlation between modes $C$ and $B_0$ as strong as possible, which from the covariance matrix pointview we want \(\langle\Delta\hat{x}_C\Delta\hat{x}_{B_0}\rangle\) and \(\langle\Delta\hat{p}_C\Delta\hat{p}_{B_0}\rangle\) to be as large as possible.
Suppose the source consists of \(n\) distinct coherent states $S_{\rho}=\bigl\{\,|\alpha_{B_0}^1\rangle,\;|\alpha_{B_0}^2\rangle,\;\dots,\;|\alpha_{B_0}^n\rangle\,\bigr\}$, and the three-mode entangled source is denoted by \(|\Psi_{ACB_0}\rangle\).
Assume the covariance matrix \(\gamma_{ACB_0}\) of \(|\Psi_{ACB_0}\rangle\) is in the standard form (so that \(\langle\Delta\hat{x}_i\Delta\hat{p}_j\rangle=0\) for all \(i,j\)), which can be arranged in the QAM case.

We take the \(x\)-quadrature as an example. Denote the mean and variance each substate \(\rho_C^i\) and of the coherent state \(|\alpha_{B_0}^i\rangle\) by
\begin{equation}\label{eq:15}
\begin{aligned}
\bar{x}_C^i &= \operatorname{Tr}\bigl[\hat{x}_C\,\rho_C^i\bigr], \quad
V_{C,x}^i = \operatorname{Tr}\bigl[\hat{x}_C^2\,\rho_C^i\bigr] - (\bar{x}_C^i)^2,\\
\bar{x}_{B_0}^i &= \langle\alpha_{B_0}^i|\hat{x}_{B_0}|\alpha_{B_0}^i\rangle, \quad
V_{B_0,x}^i = 1.
\end{aligned}
\end{equation}
The overall mean values for modes \(C\) and \(B_0\) are $\bar{x}_C=\sum_{i=1}^n p_i\,\bar{x}_C^i$ and $
\bar{x}_{B_0}=\sum_{i=1}^n p_i\,\bar{x}_{B_0}^i$, 
and the corresponding variances are
\begin{equation}\label{eq:16}
\begin{aligned}
V_{C,x} &= \sum_{i=1}^n p_i\,V_{C,x}^i \;+\; \sum_{i=1}^n p_i\bigl(\bar{x}_C^i-\bar{x}_C\bigr)^2,\\
V_{B_0,x} &= 1 \;+\; \sum_{i=1}^n p_i\bigl(\bar{x}_{B_0}^i-\bar{x}_{B_0}\bigr)^2.
\end{aligned}
\end{equation}
Given that the sub-state $\rho_{CB_0}^a$ is a product state, the covariance between \(C\) and \(B_0\) in the \(x\)-quadrature is
\begin{equation}\label{eq:17}
\begin{aligned}
&\bigl|\langle\Delta\hat{x}_C\Delta\hat{x}_{B_0}\rangle\bigr|= \bigl|\langle(\hat{x}_C-\bar{x}_C)(\hat{x}_{B_0}-\bar{x}_{B_0})\rangle\bigr|\\
&= \Biggl|\sum_{i=1}^n p_i\,
\operatorname{Tr}_{CB_0}\bigl[(\hat{x}_C-\bar{x}_C)(\hat{x}_{B_0}-\bar{x}_{B_0})
\,\rho_C^i\otimes\rho_{B_0}^i\bigr]\Biggr|\\
&= \Biggl|\sum_{i=1}^n p_i\bigl(\bar{x}_C^i-\bar{x}_C\bigr)\bigl(\bar{x}_{B_0}^i-\bar{x}_{B_0}\bigr)\Biggr|\\
&\le \sqrt{\Bigl(\sum_{i=1}^n p_i\bigl(\bar{x}_C^i-\bar{x}_C\bigr)^2\Bigr)
\Bigl(\sum_{i=1}^n p_i\bigl(\bar{x}_{B_0}^i-\bar{x}_{B_0}\bigr)^2\Bigr)}\\
&= \sqrt{\bigl(V_{C,x}-\sum_{i=1}^n p_i V_{i,Cx}\bigr)\bigl(V_{B_0,x}-1\bigr)}.
\end{aligned}
\end{equation}
The inequality above follows from the Cauchy-Schwarz inequality. Equality holds if and only if there exists a nonzero scalar \(t\) such that for all \(i\), $(\bar{x}_C^i-\bar{x}_C)/(\bar{x}_{B_0}^i-\bar{x}_{B_0}) \equiv t$.
By the uncertainty relation each sub-state satisfies \(V_{C,x}^i V_{C,p}^i\ge 1\). If we further assume \(x\) and \(p\) are symmetric for each substate (so \(V_{C,x}^i\ge 1\) and \(\sum_i p_i V_{C,x}^i\ge 1\)), then maximizing \(|\langle\Delta\hat{x}_C\Delta\hat{x}_{B_0}\rangle|\) requires \(V_{C,x}^i=1\) for all \(i\). These conditions are simultaneously met when each \(\rho_C^i\) is itself a coherent state \(|\alpha_C^i\rangle\) and the means satisfy a linear relation $\alpha_C^i/\alpha_{B_0}^i = t,\ \forall i\in\{1,\dots,n\}$.

Such a linear relationship can be realized by a beams-plitter model: prepare a coherent state \(|\beta_D^i\rangle\) with $\beta_D^i=\sqrt{1+t^2}\,\alpha_{B_0}^i$, and send it through a beams-plitter of transmittance \(\eta_{BS}=1/(1+t^2)\). In this construction Alice only needs to prepare a two‑mode entangled state of the form $|\psi_{AD}\rangle=\sum_{i=1}^n \sqrt{p_i}\,|R_A^i\rangle\,|\beta_D^i\rangle$, and then let mode \(D\) pass through the beam-splitter to produce the desired correlations between modes \(C\) and \(B_0\), as illustrated in Fig. \ref{EB}.

For the explicit design of \(|\psi_{AD}\rangle\), first symmetrize the mixed state
\(\rho_D=\sum_{i=1}^n p_i\,|\beta_D^i\rangle\langle\beta_D^i|\) in an orthonormal basis \(\{|\theta_D^i\rangle\}\):
\begin{equation}\label{eq:18}
\rho_D=\sum_{i=1}^n p_i\,|\beta_D^i\rangle\langle\beta_D^i|
=\sum_{i=1}^n v_i\,|\theta_D^i\rangle\langle\theta_D^i|.
\end{equation}
Write the orthonormal vectors as $|\theta_D^i\rangle=\sum_{j=0}^{\infty} c_{ij}\,|j\rangle$, and define
\begin{equation}\label{eq:19}
|\psi_{AD}\rangle=\sum_{i=1}^n \sqrt{v_i}\,|\varphi_A^i\rangle\,|\theta_D^i\rangle,
\end{equation}
where the states \(|\varphi_A^i\rangle\) is designed as $|\varphi_A^i\rangle=\sum_{j=0}^{\infty} (c_{ij})^*\,|j\rangle$.


\begin{figure*}[ht]
\centering
\includegraphics[width=14cm]{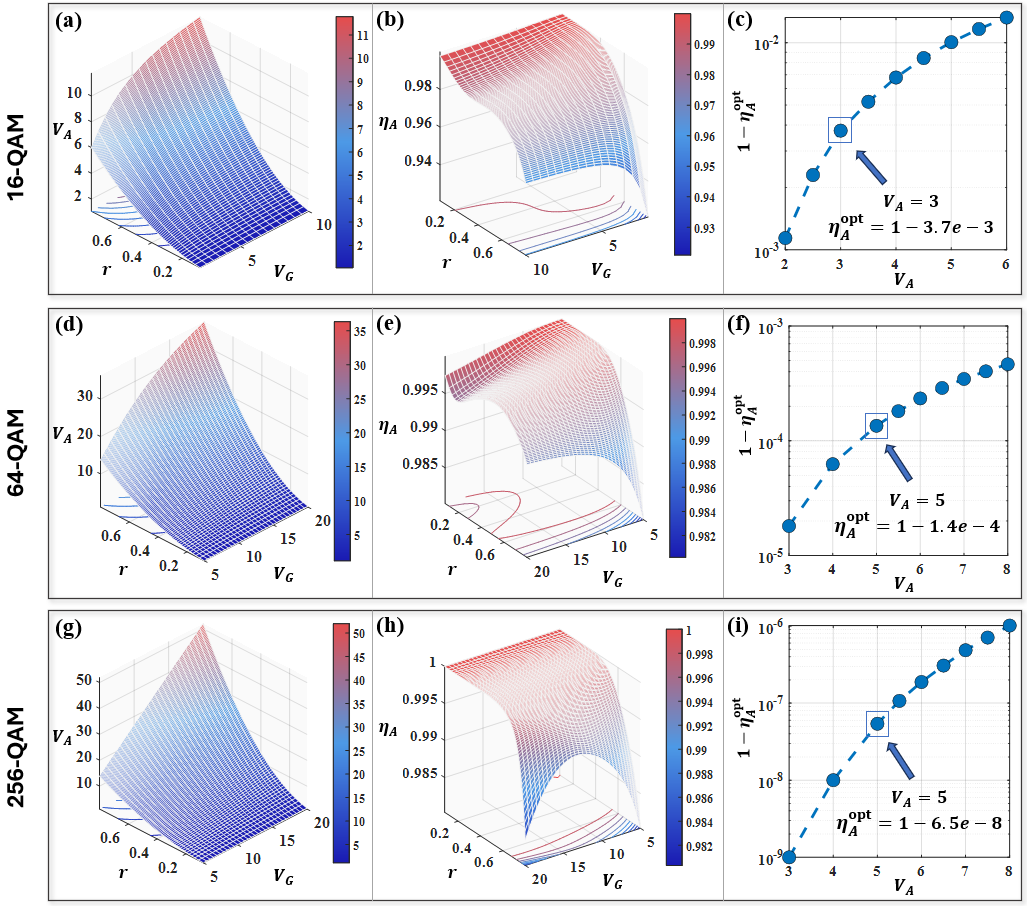}
\caption{Performance of the entangled source with 16-, 64-, 256-QAM constellation using Gaussian probability constellation shaping.}
\label{source_all}
\end{figure*}

\subsubsection{Optimization for quadrature-amplitude modulation cases}
QAM format is a standard modulation format in classical coherent communications and has been widely applied. Systems employing this modulation are naturally compatible with the existing electro-optical device.

For a standard $n$-QAM ($n=L^2$, with $L$ a positive integer) format, coherent states are placed at the intersections of \(L\) equally spaced columns and \(L\) equally spaced rows in phase space (or classically called constellation map). If the space between adjacent columns (or rows) is \(2r\), then the positions of the \(n\) coherent states can be written as
\begin{equation} 
    \{\forall \mu, \nu \in [1,L], \alpha_{\mu \nu}=(2\mu-1-L)r+i \cdot (2\nu-1-L)r\}. 
\end{equation}
For this standard QAM constellation, the covariance matrix of \(|\Psi_{AD}\rangle\) exhibits the following form,
\begin{equation} 
    \gamma_{AD}=
    \begin{bmatrix} 
            V_AI_2 & \phi_{AD}\sigma_Z \\ 
    \phi_{AD}\sigma_Z & V_DI_2 
    \end{bmatrix}, 
\end{equation}
where $I_2=
\begin{bmatrix}
		1 & 0 \\
		0 & 1
\end{bmatrix}$, $\sigma_Z=
\begin{bmatrix}
		1 & 0 \\
		0 & -1
\end{bmatrix}$. According to the uncertainty principle, we know $\phi_{AD} \le \sqrt{V_A^2-1}$, then an ideality coefficient $\eta_A=\phi_{AD}^2\Big/(V_A^2-1)$ can be defined to quantify the deviation of $|\Psi_{AD}\rangle$ from a two-mode squeezed state (TMSV). This deviation depends on three factors: the number of coherent states $n$, their positions in phase space $\alpha_{\mu \nu}$, and their probability distribution $p(\alpha_{\mu \nu})$.

Using Gaussian probability constellation shaping to optimize the distribution of these states, yields that
\begin{equation} 
    p(\alpha_{\mu \nu})=e^{\frac{-\vert\alpha_{\mu \nu}\vert^2}{(2V_G)}}\Big/\sum{e^{\frac{-|\alpha_{\mu \nu}|^2}{(2V_G)}}},
\end{equation}
where $V_G$ is a parameter related to probability shaping. 

The above processes specify the quantum state $|\Psi_{AD}\rangle$ to be transmitted.
We numerically calculate the \(\eta_A\) and $V_A$ for 16-QAM, 64-QAM and 256-QAM with different \(V_G\) and \(r\), to find a relatively optimal combination of \(V_A\) and \(\eta_A\), as shown in  Fig. \ref{source_all}. 
We first use a heatmap to illustrate the relationship between $V_A$ and \(V_G\) and \(r\) in Fig. \ref{source_all} (a) (d) and (g). Then we calculate the corresponding $\eta_A$ for given values of $V_G$ and $r$ in Fig. \ref{source_all} (b) (e) and (h), which allows us to determine the optimal $\eta_A^{\mathrm{opt}}$ for any given $V_A$, and we display the relationship of $V_A$ and $1-\eta_A^{\mathrm{opt}}$ in Fig. \ref{source_all} (b) (e) and (h). 
It can be seen that the more coherent states used and the smaller the modulation variance is set, the closer $|\Psi_{AD}\rangle$ approaches the ideal TMSV state. Considering both the performance degradation caused by the nonideal discrete-modulated source and the signal-to-noise ratio requirements of practical experiments, the typical $(V_G, V_A, \eta_A^{\mathrm{opt}})$ for three modulation formats can be selected as shown in Tab. \ref{tab1}.

\begin{table}[ht]
\centering
\begin{ruledtabular} 
\begin{tabular}{@{} l c c c @{}}
\textbf{Constellation} & \(\mathbf{V_G}\) & \(\mathbf{V_A}\) & \(\boldsymbol{\eta_A^{\mathrm{opt}}}\) \\
\hline
16-QAM & 3  & 3 & \(1-3.7\times10^{-3}\) \\
64-QAM & 6  & 5 & \(1-1.4\times10^{-4}\) \\
256-QAM & 11 & 5 & \(1-6.5\times10^{-8}\) \\
\end{tabular}
\caption{Typical $V_A$ and corresponding optimal $\eta_A^{\mathrm{opt}}$.}
\label{tab1}
\end{ruledtabular} 
\end{table}

It is worth mentioning that the quantum state sent into the channel is determined by the design of $|\psi_{AD}\rangle$, which implies that the protocol can be be customized following a specific system — one may first identify the quantum state required in practice and then optimize the corresponding two-mode entangled source. Taking the two-mode squeezed vacuum (TMSV) as an ideal reference, the gap between \(|\Psi\rangle_{AD}\) and TMSV constitutes the primary factor influencing protocol performance. Therefore, from a practical engineering perspective, we concentrate on the design of standard quadrature-amplitude-modulation (QAM) modulation source used in classical coherent communications, since systems employing this modulation are particularly well aligned with the existing electro-optic devices.

\subsection{Secret key rate calculation based on covariance matrix}\label{SKR}
Here we explain the derivation of the secret key rate formula for our three-mode protocol in Fig. \ref{EB}. It should be noted that the method is not limited to the three-mode case and can be generalized to the cases where Alice has an arbitrary $N\ (\ge 2)$ mode entangled source.

A generally used secret key rate for the asymptotic case is the Devetak--Winter formula \cite{devetak2005distillation},
\begin{equation}\label{eq:DW}\tag{3}
K=\beta I(a:b)-S(b:E),
\end{equation}
where \(I(a:b)\) is the classical mutual information between Alice and Bob, \(\beta\) is the classical reconciliation efficiency, and \(S(b:E)\) is the Holevo information \cite{holevo1973bounds} between Bob's data and the adversary Eve. Usually, \(S(b:E)\) can be replaced by any of its upper bounds, among which the upper bound \(S_{BE}^G\) derived from the Gaussian state extremality theorem \cite{garcia2006unconditional, wolf2006extremality}  is the most commonly used case, because its calculation only relies on the covariance matrix \(\gamma_{ACB}\), which can be estimated through the experimental data.

The covariance matrix \(\gamma\) of an \(N\)-mode state \(\hat{\rho}_N\) is defined as \cite{weedbrook2012gaussian}
\begin{equation}\label{eq:cov_def}\tag{4}
\gamma_{ij}:=\tfrac{1}{2}\big\langle\{\Delta\hat{r}_i,\Delta\hat{r}_j\}\big\rangle,
\end{equation}
where \(\hat{r}=\{\hat{x}_1,\hat{p}_1,\ldots,\hat{x}_N,\hat{p}_N\}\), \(\langle\hat{r}_i\rangle=\mathrm{Tr}(\hat{r}_i\hat{\rho}_N)\), and \(\Delta\hat{r}_i=\hat{r}_i-\langle\hat{r}_i\rangle\). 
Suppose \(\gamma_{ACB}\) is the covariance matrix of the state \(\rho_{ACB}\), which is the state after mode \(B_0\) of the entangled source \(\ket{\Psi_{ACB_0}}\) arriving at Bob's side through the channel. It can be represented using several sub-matrices,
\begin{equation}\label{eq:gamma_ACB}\tag{5}
\gamma_{ACB}=
\begin{pmatrix}
\gamma_A & \psi_{AC} & \kappa_{AB} \\
\psi_{AC}^T & \gamma_C & \psi_{CB} \\
\kappa_{AB}^T & \psi_{CB}^T & \gamma_B
\end{pmatrix}.
\end{equation}
Here \(\gamma_A,\gamma_C\) and \(\psi_{AC}\) can be theoretically calculated from \(\Psi_{ACB_0}\), since modes \(A\) and \(C\) are kept on Alice's side. \(\psi_{CB}\) and \(\gamma_B\) can be estimated after Alice and Bob randomly share part of their coherent measurement results. The only unknown sub-matrix is \(\kappa_{AB}\), since the measurements for modes \(A\) are not coherent measurements now.

Nevertheless, the uncertainty principle \cite{weedbrook2012gaussian} imposes a constraint on  the covariance matrix \(\gamma_N\) for an \(N\)-mode state, which is
\begin{equation}\label{eq:uncertainty}\tag{6}
\gamma_N + i\Omega_N \ge 0,
\end{equation}
where \(\Omega_N=\mathrm{diag}(\omega_1,\omega_2,\ldots,\omega_N)\), and
\begin{equation}\label{eq:omega}\tag{7}
\omega_1=\cdots=\omega_N=\omega=
\begin{pmatrix}
0 & 1\\
-1 & 0
\end{pmatrix}.
\end{equation}

We denote \(\mathcal{F}_{\kappa_{AB}}\) as the set of all \(\kappa_{AB}\) satisfying this constraint for \(\gamma_{ACB}\), that is
\begin{equation}\label{eq:S_k}\tag{8}
\mathcal{F}_{\kappa_{AB}}=\big\{\phi_{AB}\ \big|\ \gamma_{ACB}[\kappa_{AB}=\phi_{AB}]+i\Omega_3\ge 0\big\}.
\end{equation}
If \(\phi_{AB}^{\mathrm{worst}}\) is the real eavesdropping induced \(\kappa_{AB}\), then \(\phi_{AB}^{\mathrm{worst}}\in \mathcal{F}_{\kappa_{AB}}\). It can be understood that \(S_{BE}^G\) is a function of \(\kappa_{AB}\) now. Then by traversing the set \(\mathcal{F}_{\kappa_{AB}}\) for all possible \(\kappa_{AB}\), we can find the maximum of \(S_{BE}^G\). The secret key rate can be written as
\begin{equation}\label{eq:key_rate_sup}\tag{9}
K=\beta I(a:b)-\sup_{\kappa_{AB}\in \mathcal{F}_{\kappa_{AB}}} S_{BE}^G.
\end{equation}

Now we introduce the calculation method for each term. The reconciliation efficiency $\beta$ is determined by error correction codes and the signal-to-noise ratio of the measurement results. For \(\beta I(a:b)\), it can be expressed as \(\beta I(a:b)=H(b)-\mathrm{leak}_{\mathrm{EC}}\), in which \(H(b)\) is the Shannon entropy of Bob's measurement results, and \(\mathrm{leak}_{\mathrm{EC}}\) represents the information that Bob sends to Alice for the data reconciliation. Both these terms can be obtained from measured data and the error correction step. Given the channel model, the probability \(p(b\mid i)\) can be obtained from the model, which is the probability of getting Bob's measurement result \(b\) given Alice sending the state \(\rho_B^i\). The overall probability is \(p(b)=\sum_{i=1}^n p_i\,p(b\mid i)\). Since \(b\) is a discrete variable quantized from the measurement result, then
\begin{equation}\label{eq:mutual_discrete}\tag{10}
\begin{aligned}
&I(a:b)\\
&=-\sum_{b} p(b)\log_2 p(b)+\sum_{i=1}^n p_i \sum_{b} p(b\mid i)\log_2 p(b\mid i).
\end{aligned}
\end{equation}

For \(\sup_{\kappa_{AB}\in \mathcal{F}_{\kappa_{AB}}} S_{BE}^G(\kappa_{AB})\), we traverse each \(\mathcal{F}_{\kappa_{AB}}\) to calculate its corresponding \(S_{BE}^G(\kappa_{AB})\) and find the maximal value. \(S_{BE}^G(\kappa_{AB})\) be expressed as $S_{BE}^G(\kappa_{AB})= S\bigl(\rho_{A B C}^G\bigr)-S\bigl(\rho_{A C\mid b}^G\bigr)$,
where \(\rho_{A C B}^G\) is the Gaussian state with the same covariance matrix as \(\gamma_{ACB}(\kappa_{AB})\), and \(\rho_{A C\mid b}^G\) denotes the conditional Gaussian state (with covariance \(\gamma_{AC\mid b}\)) related to Bob's measurement method. The procedures to obtain \(\gamma_{AC\mid b}\) from \(\gamma_{ACB}\) and to evaluate von Neumann entropies for Gaussian states can be found in \cite{weedbrook2012gaussian, fossier2009improvement}.

\begin{figure}
    \includegraphics[width= 8.5 cm]{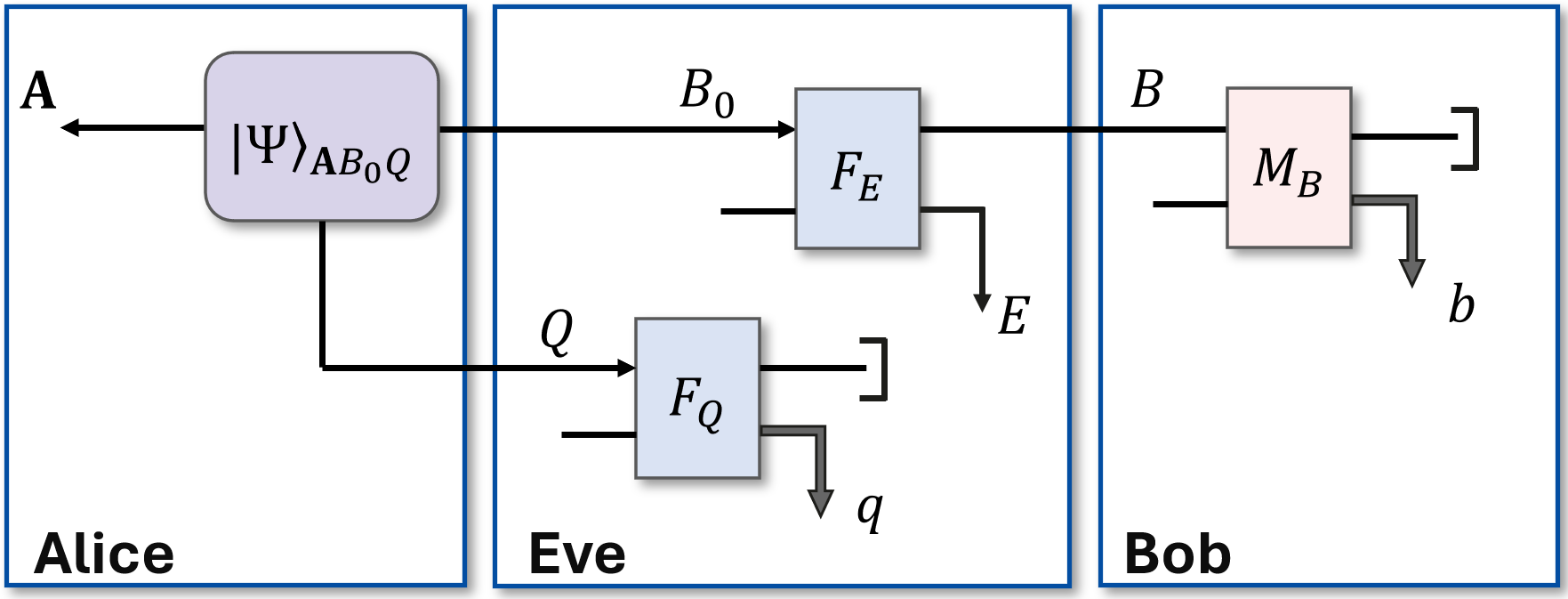}
    \caption{Mixture and measurement of two entangled sources.
    }\label{Mix}
\end{figure}

\subsection{Convexity of secret key rate}\label{Convexity}
Here we describe the convexity of SKR as the following lemma.
\begin{lemma}
Let $\rho_{\mathbf{A}B}^{(+)}$ and $\rho_{\mathbf{A}B}^{(-)}$ be two $N$-mode quantum states shared by Alice and Bob, where $\mathbf{A}=A_1A_2\cdots A_{N-1}$ denotes the modes held by Alice and $B$ denotes Bob's mode. Assume that the two states yield the same SKR,
\begin{equation}
K\left(\rho_{\mathbf{A}B}^{(+)}\right)=
K\left(\rho_{\mathbf{A}B}^{(-)}\right)
=K_0.
\end{equation}
For any probabilistic mixture
\begin{equation}
\rho_{\mathbf{A}B}=
p_1\rho_{\mathbf{A}B}^{(+)}+
p_2\rho_{\mathbf{A}B}^{(-)},
\quad
p_1+p_2=1,
\end{equation}
the SKR of the mixed state satisfies
\begin{equation}
K(\rho_{\mathbf{A}B})
\le p_1K\left(\rho_{\mathbf{A}B}^{(+)}\right)+
p_2K\left(\rho_{\mathbf{A}B}^{(-)}\right)=
K_0.
\end{equation}
It can be expressed as
\begin{equation} 
    K\big(\sum_{i=1}^2 p_i \rho_i\big) \le \sum_{i=1}^2 p_iK(\rho_i)=K(\rho_i),
\end{equation} 
where $\sum_i p_i=1$, $\rho_1=\rho_{\mathbf{A}B}^{(+)}$, $\rho_2=\rho_{\mathbf{A}B}^{(-)}$.
\end{lemma}

\begin{proof}
We first explain the steps for mixture of the two states, as shown in Fig. \ref{Mix}:

1). Suppose there are two entanglement sources \(|\Psi\rangle_{\mathbf{A}B_0}^{(+)}\) and \(|\Psi\rangle_{\mathbf{A}B_0}^{(-)}\). After the optimal Gaussian attack by the eavesdropper, they yield the target states $\rho_{\mathbf{A}B}^{(+)}$ and $\rho_{\mathbf{A}B}^{(-)}$. 
Then mix the two entanglement sources and get \(\rho_{\mathbf{A}B_0}=p\rho_{\mathbf{A}B_0}^{(+)}+(1-p)\rho_{\mathbf{A}B_0}^{(-)}\).

2). Introduce an auxiliary system \(Q\), which can purify the entire system and yield 
\begin{equation}
    |\Psi\rangle_{\mathbf{A}B_0Q}=\frac{1}{\sqrt{2}}\bigl(|\Psi\rangle_{\mathbf{A}B_0}^{(+)}|0\rangle_Q+|\Psi\rangle_{\mathbf{A}B_0}^{(-)}|1\rangle_Q\bigr),
\end{equation}
and this is a purification of the mixed state $\rho_{\mathbf{A}B_0}$. Then a POVM measurement $F_Q$ is performed on \(Q\), obtaining the result $q$, with the result $q=|0\rangle_Q$ projecting the subsystem \(ACB_0\) onto \(|\Psi\rangle_{\mathbf{A}B_0}^{(+)}\), and $q=|1\rangle_Q$ projecting it onto $|\Psi\rangle_{\mathbf{A}B_0}^{(-)}$.

3). The mixed state $\rho_{B_0}$ is sent to Bob  through the quantum channel. Eve implement an attack described by a unitary $F_E$ and can purify the system $\rho_{\mathbf{A}B Q}$, so that the entire state after the channel is
\begin{equation}
|\Psi\rangle_{\mathbf{A}B E Q}=U_{EB}\bigl(|\Psi\rangle_{\mathbf{A}B_0 Q}\otimes|\Psi\rangle_{E}\bigr),
\end{equation}
and we can get the mixed state $\rho_{\mathbf{A}B}=p\rho_{\mathbf{A}B}^{(+)}+(1-p)\rho_{\mathbf{A}B}^{(-)}$.

4). Bob receive the state $\rho_B$ and perform a heterodyne measurement, obtaining the result $b$.

5). The measurement $q$ is published, allowing Alice and Bob to align which state they have sent and measured, and then perform post-processing.

Now we analyze the maximum information Eve can obtain. Suppose Eve can fully obtain the result $q$ and get the joint system $Eq$, then Eve's accessible information should be $S(b:Eq)$. From the definition of von Neumann mutual information, we have,
\begin{equation}\label{eq:mutual_def}
    \begin{aligned}
        S(b:E q)=&S(b)+S(E q)-S(b q E) \\
        =&S(\rho_b)+S(\rho_{E q})-S(\rho_{b q E}).
    \end{aligned}
\end{equation}
The calculation of \(S(b:E q)\) is based on the following classical-quantum state
\begin{equation}\label{eq:rho_bqE}
\rho_{b q E}=\iint \mathrm{d}b\,\mathrm{d}q\; p(b,q)\,|b\rangle\langle b|\otimes|q\rangle\langle q|\otimes\rho_E^{b q}.
\end{equation}
From this we can write
\begin{equation}\label{eq:rho_b}
    \begin{aligned}
        \rho_b=&\int \mathrm{d}b\; p(b)\,|b\rangle\langle b|,\\
\rho_{E q}=&\int \mathrm{d}q\; p(q)\,|q\rangle\langle q|\otimes\rho_E^{q},\\
\rho_{E q}^{b}=&\int \mathrm{d}q\; p(q\mid b)\,|q\rangle\langle q|\otimes\rho_E^{b q}.
    \end{aligned}
\end{equation}
Using von Neumann entropy and joint entropy definitions, we have
\begin{equation}\label{eq:S_b}
    \begin{aligned}
        S(b)=&S\bigl(\int \mathrm{d}b\; p(b)\,|b\rangle\langle b|\bigr)=H\bigl(p(b)\bigr),\\
S(E q)=&H\bigl(p(q)\bigr)+\int \mathrm{d}q\; p(q)\,S\bigl(\rho_E^{q}\bigr),\\
S(b q E)=&H\bigl(p(b)\bigr)+\int \mathrm{d}b\; p(b)\,S\bigl(\rho_{E q}^{b}\bigr).
    \end{aligned}
\end{equation}
The conditional entropy can be written as
\begin{equation}\label{eq:cond_entropy}
    \begin{aligned}
        \int \mathrm{d}b\; p(b)\,S\bigl(\rho_{E q}^{b}\bigr)=&\int \mathrm{d}b\; p(b)\,H\bigl(p(q\mid b)\bigr)+\\
&\int \mathrm{d}b\; p(b)\int \mathrm{d}q\; p(q\mid b)\,S\bigl(\rho_E^{b q}\bigr).
    \end{aligned}
\end{equation}
Therefore, \(S(b:E q)\) can be written as
\begin{equation}\label{eq:S_bEq_expand}
\begin{aligned}
S(b:E q)=&H\bigl(p(q)\bigr)+\int\mathrm{d}q\; p(q)\,S(\rho_E^{q})\\
&-\int\mathrm{d}b\; p(b)\,S(\rho_{E q}^{b}).
\end{aligned}
\end{equation}
SInce \(A B C E Q\) is a pure system, for rank-1 measurements we have$S(\rho_E^{q})=S(\rho_{\mathbf{A}B}^{q})$ and $S(\rho_E^{b q})=S(\rho_{\mathbf{A}}^{b q})$. Moreover, since $H(p(q))=H(q)$, $H(q\mid b)=\int \mathrm{d}b\; p(b)\,H\bigl(p(q\mid b)\bigr)$, $I(b:q)=H(q)-H(q\mid b)$. Using these relations, we obtain a simplified form of Eve's knowledge
\begin{equation}\label{eq:S_bEq_simplified}
\begin{aligned}
S(b:E q)
&=I(b:q)+\int\mathrm{d}q\; p(q)\Bigl[S(\rho_{\mathbf{A}B}^{q})\\
&-\int\mathrm{d}b\; p(b\mid q)\,S(\rho_{\mathbf{A}}^{b q})\Bigr].
\end{aligned}
\end{equation}
The bracketed term has the same form as a classical--quantum mutual information, it can be written as
\begin{equation}\label{eq:cond_mutual}
S(E:b\mid q)= S(\rho_{\mathbf{A}B}^{q})-\int\mathrm{d}b\; p(b\mid q)\,S(\rho_{\mathbf{A}}^{b q}).
\end{equation}
When the measurement result $q$ is \(q=|0\rangle_Q\) we have
\begin{equation}\label{eq:cond0}
S(b:E\mid q=|0\rangle_Q)=S(b:E)\bigl(\rho_{\mathbf{A}B}^{(+)}\bigr),
\end{equation}
and when \(q=|1\rangle_Q\) we have
\begin{equation}\label{eq:cond1}
S(b:E\mid q=|1\rangle_Q)=S(b:E)\bigl(\rho_{\mathbf{A}B}^{(-)}\bigr).
\end{equation}
Since $S(E:b)\bigl(\rho_{\mathbf{A}B}^{+}\bigr)=S(E:b)\bigl(\rho_{\mathbf{A}B}^{-}\bigr)=S(E:b)_1$, and $S(E:b)_1$ is defined as Eve's information obtained from either independent quantum state. Then $S(b:E q)$ can be further simplified as
\begin{equation}\label{eq:S_bEq_final}
\begin{aligned}
S(b:E q)
&=I(b:q)+pS(E:b)\bigl(\rho_{\mathbf{A}B}^{((+))}\bigr)+\\
&(1-p)S(E:b)\bigl(\rho_{\mathbf{A}B}^{((-))}\bigr)\\
&=I(b:q)+S(E:b)\bigl(\rho_{\mathbf{A}B}^{(+)}\bigr)\ge S(E:b)_1.
\end{aligned}
\end{equation}
Thus, the convexity of Eve's information is proven.

Our goal is to using the covariance matrix $\gamma_{\mathbf{A}B}$ of the mixed state $\rho_{\mathbf{A}B}$ to calculate the mutual information \(S(b:E Q)_{\mathrm{G}}\) when the system $\mathbf{A}BQE$ is Gaussian, and with the quantum data-processing inequality \cite{schumacher1996quantum} and Gaussian optimality \cite{garcia2006unconditional, wolf2006extremality}, we have
\begin{equation}\label{eq:QDPI}
S(b:E Q)_{\mathrm{G}} \ge S(b:E Q)\ge S(b:E q)\ge S(E:b)_1.
\end{equation}

Additionally, the mutual information between Alice and Bob is a classical result that is not affected by quantum operations. Therefore, the relationship between SKR is 
\begin{equation}\label{eq:QDPI}
K(\gamma_{\mathbf{A}B})=K(\rho_{\mathbf{A}B}^G) \le K(\rho_{\mathbf{A}B})\le K(\rho_{\mathbf{A}B}^{(+)}).
\end{equation}
Finally, the convexity of SKR is proved.
\end{proof}

\subsection{Covariance matrix symmetrization}\label{symmetrization}

In this subsection, we detail the symmetrization steps for covariance matrix. Firstly, suppose the initial three-mode quantum state is $\rho_{ACB}^0$, with covariance matrix $\gamma_{ACB}^{(0)}\ (=\gamma_{ACB})$, which can be written as
\begin{equation} 
    \gamma_{ACB}^{(0)}=\gamma_{ACB}=
    \begin{pmatrix} 
    V_{A}I_{2} & \phi_{AC}\,\sigma_{Z} & \kappa_{AB}^{(0)} \\ 
    \phi_{AC}\,\sigma_{Z} & V_CI_{2} & \psi _{CB}^{(0)}\\
    (\kappa_{AB}^{(0)})^T & (\psi _{CB}^{(0)})^T &\gamma_B^{(0)}
    \end{pmatrix},
\end{equation} 
where the upper-left $4\times 4$ sub-matrix $\gamma_{AC}$ is symmetric due to the intrinsic symmetry of the designed source, where 
\begin{equation}
    \begin{aligned}
    V_C=&(1-\eta_{BS})(V_A-1)+1, \\
    \phi _{AC}=&\sqrt{(1-\eta_{BS})\eta_A(V_A^2-1)},
    \end{aligned}
\end{equation}
and $I_2= (1\ 0\ ;\ 0\ 1)$, $\sigma_Z=(1\ 0\ ;\ 0\ -1)$. The remaining sub-matrices are estimated using the measurement data, resulting in an unknown matrix $\kappa_{AB}$, and two non-symmetric matrices $\psi_{CB}$ and $\gamma_{B}$,
\begin{equation}
    \begin{aligned}
    \kappa_{AB}^{(0)}&=(\kappa_{11}\ \kappa_{12}\ ;\ \kappa_{21}\ \kappa_{22}),\\
    \psi _{CB}^{(0)}&=(\phi_{11}\ \phi_{12}\ ;\ \phi_{21}\ \phi_{22}),\\
    \gamma_B^{(0)}&=(b_{11}\ b_{12}\ ;\ b_{12}\ b_{22}).
    \end{aligned}
\end{equation}
Then \(\gamma_{ACB}^{(0)}\) is symmetrized by the three steps as below:


\textit{$\mathbf{S}_1$: Controlled rotation.}
The first step symmetrizes the block $\gamma_B^{(0)}$. We introduce an auxiliary system $G$ as a control quantum state $|\varphi \rangle_G = \sqrt{1/2}(|0 \rangle_G+|1 \rangle_G)$, whose controlled operation $U_G=|0\rangle\langle0|_G\otimes I_{ABC} 
+|1\rangle\langle1|_G\otimes R_A\!\left(\tfrac{\pi}{2}\right)\otimes R_{BC}\!\left(-\tfrac{\pi}{2}\right)$ is applied to the target state $\rho_{ACB}^{(0)}$.
Then the global state becomes 
\begin{equation}
    \begin{aligned}
        \rho_{ACBG}^{(1)}=U_G\big(\rho_{ACB}^{(0)}\otimes|\varphi\rangle_G \langle \varphi|\big)U_G^{\dagger}.
    \end{aligned}
\end{equation} 
If \(|\varphi \rangle_G=|0\rangle_G\), the target state $\rho_{ACB}^{(0)}$ is left alone; if \(|\varphi \rangle_G=|1\rangle_G\), mode \(A\) is rotated by \(90^\circ\) while \(B\) and \(C\) are rotated by \(-90^\circ\), yielding the state $\rho_{ACB}^{\mathrm{(0R)}}$ with covariance matrix $\gamma_{ACB}^{\mathrm{(0R)}}$. After a quantum measurement on $G$, we can get the desired state $\rho_{ACB}^{(1)}$ and its covariance matrix
\begin{equation}
   \rho_{ACB}^{(1)}=\frac{\rho_{ACB}^{(0)}+\rho_{ACB}^{\mathrm{(0R)}}}{2},\ \gamma_{ACB}^{(1)}=\frac{\gamma_{ACB}^{(0)}+\gamma_{ACB}^{\mathrm{(0R)}}}{2},
\end{equation} 
where the sub-matrix $\gamma_{B}^{(1)}$ takes a symmetric form as 
\begin{equation}
   \gamma_{B}^{(1)}=V_BI_2=
   \begin{pmatrix}
    \frac{b_{11}+b_{22}}{2} & 0 \\ 
    0 & \frac{b_{11}+b_{22}}{2}
   \end{pmatrix}.
\end{equation} 
The remaining sub-matrices are
\begin{equation}
\kappa_{AB}^{(1)}=\begin{pmatrix}
\kappa^{(1)} & \Delta\kappa^{(1)} \\
\Delta\kappa^{(1)} & -\kappa^{(1)}
\end{pmatrix},
\ 
\psi _{CB}^{(1)}=\begin{pmatrix}
\phi^{(1)} & \Delta\phi^{(1)} \\
-\Delta\phi^{(1)} & \phi^{(1)}
\end{pmatrix},
\end{equation}
where $\kappa^{(1)}=(\kappa_{11}-\kappa_{22})/2$, $\Delta\kappa^{(1)}=(\kappa_{12}+\kappa_{21})/2$, $\phi^{(1)}=(\phi_{11}+\phi_{22})/2$, $\Delta\phi^{(1)}=(\phi_{12}-\phi_{21})/2$. It's obvious that $K(\rho_{ACB}^{(0)})=K(\rho_{ACB}^{\mathrm{(0R)}})$ because the rotation is a unitary operation. So based on the convexity of SKR, we have $K\big(\rho_{ACB}^{(1)}\big) \le K\big(\rho^{(0)}_{ACB}\big)$.

\textit{$\mathbf{S}_2$: Local rotation.}
This step symmetrizes the block $\psi _{CB}^{(2)}$. A local unitary operation $R_B$ with a symplectic matrix $X_B=(\cos \theta\ -\sin \theta\ ;\ \sin \theta\ \cos \theta)$ is applied to mode $B$, yielding
\begin{equation}
    \begin{aligned}
          \rho_{ACB}^{(2)}=&(I_{AC} \otimes R_B)\rho_{ACB}^{(1)}(I_{AC} \otimes R_B)^{\dagger },\\ 
          \gamma_{ACB}^{(2)}=&(I_{4} \oplus X_B)\gamma_{ACB}^{(1)}(I_{4} \oplus X_B)^T.
    \end{aligned}
\end{equation}
We let $\tan \theta= (\Delta\phi^{(1)})/(\phi^{(1)})$, then
\begin{equation}
    \begin{aligned}
        &\psi _{CB}^{(2)}=X_B\psi _{CB}^{(1)}X_B^T=\phi_{CB}I_2 = \\
     &\begin{pmatrix}
    \sqrt{(\phi^{(1)})^2 + (\Delta\phi^{(1)})^2} & 0 \\ 
    0 & \sqrt{(\phi^{(1)})^2 + (\Delta\phi^{(1)})^2}
   \end{pmatrix}.
    \end{aligned}
\end{equation}
Similarly,
\begin{equation}
  \kappa _{AB}^{(2)}=X_B\kappa _{AB}^{(1)}X_B^T=\begin{pmatrix}
\kappa^{(2)} & \Delta\kappa^{(2)} \\
\Delta\kappa^{2} & -\kappa^{(2)}
\end{pmatrix},  
\end{equation} 
where $\kappa^{(2)}=\kappa^{(1)} \cos \theta + \Delta\kappa^{(1)} \sin \theta$, and $\Delta\kappa^{(2)}=-\kappa^{(1)} \sin \theta + \Delta\kappa^{(1)} \cos \theta$.
This unitary operation does not change SKR, therefore we have $K\big(\rho_{ACB}^{(2)}\big) = K\big(\rho_{ACB}^{(1)}\big)$.

\textit{$\mathbf{S}_3$: State mixture.}
This step is to symmetrize the last block $\kappa _{AB}^{(2)}$. 
Let \(\rho_{ACB}^{(G,2+)}\) denote the Gaussian state sharing the same covariance matrix \(\gamma_{ACB}^{(2+)}\) (\(=\gamma_{ACB}^{(2)}\)) with \(\rho_{ACB}^{(2)}\). Define another Gaussian state \(\rho_{ACB}^{(G,2-)}\) associated with the covariance matrix \(\gamma_{ACB}^{(2-)}\), obtained from \(\gamma_{ACB}^{(2+)}\) by replacing \(\Delta\kappa^{(2)}\) with \(-\Delta\kappa^{(2)}\) while leaving all other entries unchanged. The corresponding symmetrized state and covariance matrix are
\begin{equation}
\rho_{ACB}^{\mathrm{sym}}=
\frac{\rho_{ACB}^{(G,2+)}+\rho_{ACB}^{(G,2-)}}{2},
\quad
\gamma_{ACB}^{\mathrm{sym}}=
\frac{\gamma_{ACB}^{(2+)}+\gamma_{ACB}^{(2-)}}{2},
\end{equation}
where the sub-matrix $\kappa _{AB}^{\mathrm{sym}}$ is symmetric and
\begin{equation}
    \kappa _{AB}^{\mathrm{sym}}=\kappa \sigma_Z= 
    \begin{pmatrix}
        \kappa^{(2)} & 0 \\
        0 & -\kappa^{(2)}
    \end{pmatrix}.
\end{equation}
Following the convexity of SKR, it can be proved that $K\big(\rho_{ACB}^{(3)}\big) \le K\big(\rho_{ACB}^{(2)}\big)$. 
Since \(\gamma_{ACB}^{(2+)}\) and \(\gamma_{ACB}^{(2-)}\) have the same symplectic spectra, the corresponding Gaussian states satisfy $K\left(\rho_{ACB}^{(G,2+)}\right)=K\left(\rho_{ACB}^{(G,2-)}\right)$.
Following the convexity of SKR and the Gaussian optimality theorem, $K\big(\rho_{ACB}^{\mathrm{sym}}\big) \le K\big(\rho_{ACB}^{(G,2+)}\big) \le K\big(\rho_{ACB}^{(2)}\big)$.

To obtain the mixed state, we first equally mix two entangled sources. Consider a physical source state \(|\Psi\rangle_{ACB_0}^{(2+)}\), which evolves into \(\rho_{ACB}^{(2)}\) after transmission through the quantum channel controlled by Eve. By Gaussian optimality, \(\rho_{ACB}^{(2)}\) can be replaced in the security analysis by the Gaussian state \(\rho_{ACB}^{(G,2+)}\) sharing the same covariance matrix \(\gamma_{ACB}^{(2+)}\). 
Since \(\gamma_{ACB}^{(2+)}\) and \(\gamma_{ACB}^{(2-)}\) have the same symplectic eigenvalues, Williamson's theorem \cite{GaussianQuantumInformation} guarantees that they are related by a symplectic transformation. Therefore, \(\rho_{ACB}^{(G,2-)}\) can be obtained from \(\rho_{ACB}^{(G,2+)}\) through the corresponding Gaussian unitary transformation and is also physical. Consequently, there exists a physical source state \(|\Psi\rangle_{ACB_0}^{(2-)}\), whose Gaussian counterpart after transmission is \(\rho_{ACB}^{(G,2-)}\).

Next, introduce another auxiliary system \(Q\), which can purify the entire system and yield 
$|\Psi\rangle_{ACB_0Q}^{(3)}=\frac{1}{\sqrt{2}}\bigl(|\Psi\rangle_{ACB_0}^{(2+)}|0\rangle_Q+|\Psi\rangle_{ACB_0}^{(2-)}|1\rangle_Q\bigr)$.
Then \(Q\) is measured and the result $q$ is published, with the result $q=|0\rangle_Q$ projecting the subsystem \(ACB_0\) onto \(|\Psi\rangle_{ACB_0}^{(2+)}\), and $q=|1\rangle_Q$ projecting it onto $|\Psi\rangle_{ACB_0}^{(2-)}$. After the measurement and Eve's attack, we can get the mixed state.

Ultimately, we obtain a scaling chain of SKR 
\begin{equation}
        K\big(\rho_{ACB}^{\mathrm{sym}}\big) \le K\big(\rho_{ACB}^{(2)}\big)
        =K\big(\rho_{ACB}^{(1)}\big) \le K\big(\rho^{(0)}_{ACB}\big),
\end{equation}
which demonstrates that we can safely calculate the SKR with a symmetric matrix, thus reducing the problem to solving for the single parameter $\kappa$. 

The above method is not limited to the case where the block $\gamma_{AC}$ is already in the symmetric form. More generally, an arbitrary two-mode Gaussian state can always be transformed into the standard form by local linear unitary Bogoliubov operations (LLUBOs) \cite{duan2000inseparability}. Since LLUBOs are local Gaussian unitary transformations, they leave the SKR invariant. Then the proposed three-step symmetrization procedure can be directly applied to general three-mode covariance matrices.

\subsection{Analytical solution for the range of $\kappa$}\label{kappa}
After we symmetrize the three-mode covariance matrix $\gamma_{ACB}$ and get $\gamma_{ACB}^{\mathrm{sym}}$, we reduce the problem to solving the single parameter $\kappa$. Applying the uncertainty-principle constraint $\gamma_{ACB}^{\mathrm{sym}}+i\Omega \ge 0$, the tightest feasible interval for $\kappa$ is obtained from the sixth-order principal minor, yielding $\kappa \in \mathcal{F}_{\kappa}=[\bar{\kappa}-R, \bar{\kappa}+R]$, where $\bar{\kappa} = \sqrt{T_{cb}\,\eta_{A}\,(1-\eta_{BS})\,\big(V_{A}^{2}-1\big)}$, and $R = \sqrt{T_{cb}\,\varepsilon_{cb}\,(1-\eta_{A})\,(V_{A}+1)}$.
The parameters \(T_{cb}\) and \(\varepsilon_{cb}\) are defined from the symmetrized two-mode covariance matrix \(\gamma_{CB}^{\mathrm{sym}}\) of modes \(B\) and \(C\) to provide an equivalent description of the channel transmittance and excess noise, where
\begin{equation}\label{eq:gamma_CB}
\gamma_{CB}^{\mathrm{sym}}=\begin{pmatrix} V_C I_2 & \phi_{CB} I_2 \\ \phi_{CB} I_2 & V_B I_2 \end{pmatrix}.
\end{equation}

For a matrix having the following form (assuming \(a>b\))
\begin{equation}\label{eq:gamma}
\gamma=\begin{pmatrix} aI_2 & cI_2 \\ cI_2 & bI_2 \end{pmatrix},
\end{equation}
the uncertainty principle requires $\det\bigl(\gamma+i\Omega\bigr)\ge 0$, which can be solved as
\begin{equation}\label{eq:c_upper}
c^2\le (a-1)(b-1)\le (a-1)^2.
\end{equation}
So we can define the channel transmittance and excess noise as
\begin{equation}\label{eq:T_def}
T_{\mathrm{ch}}=\frac{c^2}{(a-1)^2},\quad T\le 1, \quad \varepsilon_{\mathrm{ch}}=\frac{b-1}{T}-(a-1).
\end{equation}

Similarly, for the two-mode covariance matrix \(\gamma_{CB}^{\mathrm{sym}}\), we can define
\begin{equation}\label{eq:Tcb_epscb}
T_{cb}=\frac{\phi_{CB}^2}{(V_C-1)^2},\qquad
\varepsilon_{cb}=\frac{V_B-1}{T_{cb}}-(V_C-1).
\end{equation}
With these two parameters, we can obtain the analytical expression for the range of $\kappa$ as above. 

If we use the channel transmittance $T$ and excess noise $\varepsilon$ defined in Fig. \ref{EB}, as a linear channel model for performance simulation, then $\bar{\kappa}$ and $R$ can be re-expressed as $\bar{\kappa} = \sqrt{T\eta_{A}\eta_{BS}\,\big(V_{A}^{2}-1\big)}$, and $R = \sqrt{T\varepsilon(1-\eta_{A})(V_{A}+1)}$. And the correspondence between the parameters under these two definitions is
\begin{equation}\label{eq:Tcb_epscb}
T_{cb}=\frac{\eta_{BS}}{1-\eta_{BS}}T,\qquad
\varepsilon_{cb}=\frac{1-\eta_{BS}}{\eta_{BS}}\varepsilon.
\end{equation}

As for finding an analytical solution for $\kappa^{\mathrm{worst}}$ that maximizes $S_{BE}^G$, the problem is quite complex. Here, we present an avenue for exploration. 
Treating \(\kappa\) as a perturbative parameter provides a possible route to characterize the local dependence of \(S_{BE}^G\) on \(\kappa\). The symplectic eigenvalues can be expanded perturbatively around a reference point and substituted into \(S_{BE}^G\) through the chain rule. The resulting expansion can then be used to analyze the stationary condition \(dS_{BE}^G/d\kappa=0\), yielding an approximate analytical equation for the stationary points of \(S_{BE}^G\). However, since the perturbative coefficients depend on the full symplectic spectrum of the covariance matrix, obtaining a general closed-form expression for the maximizing value of \(\kappa\) remains highly challenging. Developing such an analytical solution is therefore left for future investigation.
\begin{figure}[t]
    \includegraphics[width= 8.5 cm]{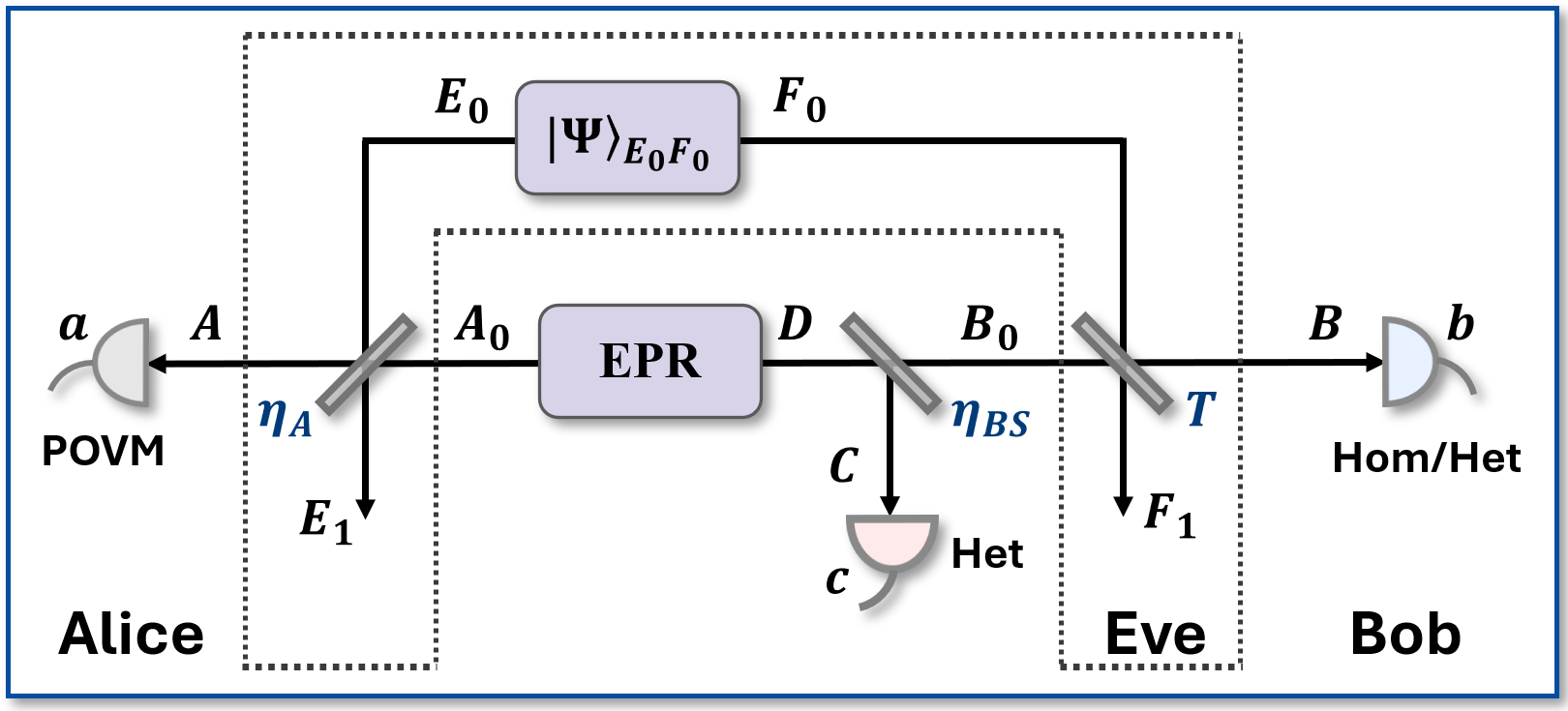}
    \caption{The attack model for the discrete-modulated protocol, in which the non-ideality of the discrete-modulated source is modeled as Eve's attack. 
    Alice prepares an Einstein-Podolski-Rosen (EPR) source and mode $D$ is divided by a beam splitter with transmittance $\eta_{BS}$ and get modes $C$ and $B_0$. Eve's attack is modeled as a two-mode state $|\Psi\rangle_{E_0F_0}$, where mode $E_0$ is coupled into the system via a beam splitter with transmittance $\eta_A$, mode $F_0$ is coupled into the system via a beam splitter with transmittance $T$. 
    }\label{Attackmodel}
\end{figure}

\subsection{Attack model}\label{model}

In this subsection, we provide a model to explain the unusual phenomenon that SKR exhibits nonlinear variations within the interval $\mathcal{F}_{\kappa}$. The model is presented in Fig. \ref{Attackmodel}, where non-ideality of the discrete-modulated source is modeled as an attack by Eve. Specifically, Alice prepares an Einstein-Podolski-Rosen (EPR) state, where mode $D$ is divided by a beam splitter with transmittance $\eta_{BS}$ and get modes $C$ and $B_0$. Eve's attack is modeled as a two-mode state $|\Psi\rangle_{E_0F_0}$, where mode $E_0$ is coupled into the system via a beam splitter with transmittance $\eta_A$, mode $F_0$ is coupled into the channel via a beam splitter with transmittance $T$. After Eve's, attack, mode $A$ and mode $B$ are obtained. Then, Alice uses POVM for mode $A$ and heterodyne measurement for mode $C$. Mode $B$ is homodyne or heterodyne measured by Bob.

According to this model, we can write the relationships between the different modes as follows:
\begin{equation}
    \begin{aligned}
        A =& \sqrt{\eta_A}A_0+\sqrt{1-\eta_A}E_0, \\
        B_0 =&\sqrt{\eta_{BS}}D_0+\sqrt{1-\eta_{BS}}N, \\
        B=& \sqrt{T}B_0+\sqrt{1-T}F_0,
    \end{aligned}
\end{equation}
where $N$ represents the vacuum state with variance $V_N=1$. The variance of mode $E_0$ is $V_{E_0}=V_{A}$. Mode $F_0$ represents the excess noise injected into the channel with variance $V_{F_0}=1+T\varepsilon /(1-T)$.

Thus, the covariance between mode $A$ and mode $B$ can be numerically calculated as 
\begin{equation}
    \begin{aligned}
        \phi_{AB} =& \sqrt{\eta_AT}\langle A_0B_0\rangle+\sqrt{1-\eta_A}\sqrt{1-T}\langle E_0F_0\rangle, \\
        =&\sqrt{\eta_AT\eta_{BS}}\sqrt{V_A^2-1}+\sqrt{1-\eta_A}\sqrt{1-T}\langle E_0F_0\rangle.
    \end{aligned}
\end{equation}
According to the uncertainty principle, we have
\begin{equation}
    \begin{aligned}
        |\langle E_0F_0\rangle| \le & \sqrt{(V_A+1)(V_{F_0}-1)} \\
        = & \sqrt{(V_A+1)(1+\frac{T\varepsilon}{1-T})},
    \end{aligned}
\end{equation}
When the covariance $\langle E_0F_0\rangle $ takes its minimum negative value, the two-mode state $|\Psi\rangle_{E_0F_0}$ is an ideal entangled state. At this point, the covariance between modes $A$ and $B$ is
\begin{equation}
        \phi_{AB}^{\mathrm{min}} =\sqrt{\eta_AT\eta_{BS}}\sqrt{V_A^2-1}-\sqrt{(1-\eta_A)T\varepsilon (V_A+1)}.
\end{equation}

Interestingly, the value of $\phi_{AB}^{\mathrm{min}}$ is exactly equal to the left boundary of the interval, i.e., $\phi_{AB}^{\mathrm{min}}=\bar{\kappa}-R$. 
This provides some insight on the non-monotonic behavior of the secret key rate, which is not necessarily minimized at the left boundary under different conditions.
As $\phi_{AB}$ changes from the mean value to the left boundary, Eve's two mode state $\phi_{E_0F_0}$ evolves from a mixed tensor product state to an entangled state, with its degree of entanglement reaching a maximum at the left boundary. While not rigorously proven, numerical observations suggest that, when SKR exhibits a non-monotonic dependence on $\kappa$, the maximum value of $S_{BE}^G$ happens near the point where $\phi_{E_0F_0}$ becomes entangled.
Since multipartite entanglement exhibits complex and potentially non-monotonic behavior, Eve's accessible information may not be a monotonic function of $\kappa$. The underlying mechanism remains unclear and deserves further investigation.

\begin{figure}
    \centering
    \includegraphics[width= 8.5 cm]{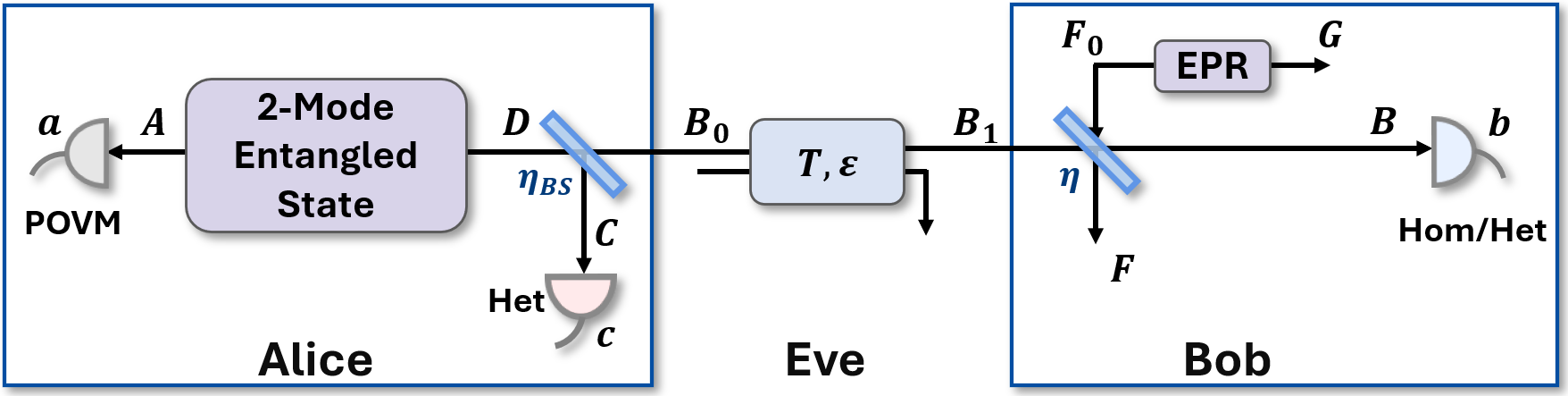}
    \caption{The Entanglement-based (EB) scheme of the proposed DM CV-QKD protocol with detector modeling. POVM, positive-operator valued measurement; Hom, homodyne detection; Het, heterodyne measure. The practical detector is characterized by the detection efficiency $\eta$ and electronic noise $v_el$.
    }\label{Detector}
\end{figure}

\begin{figure*}[t]
    \includegraphics[width=17 cm]{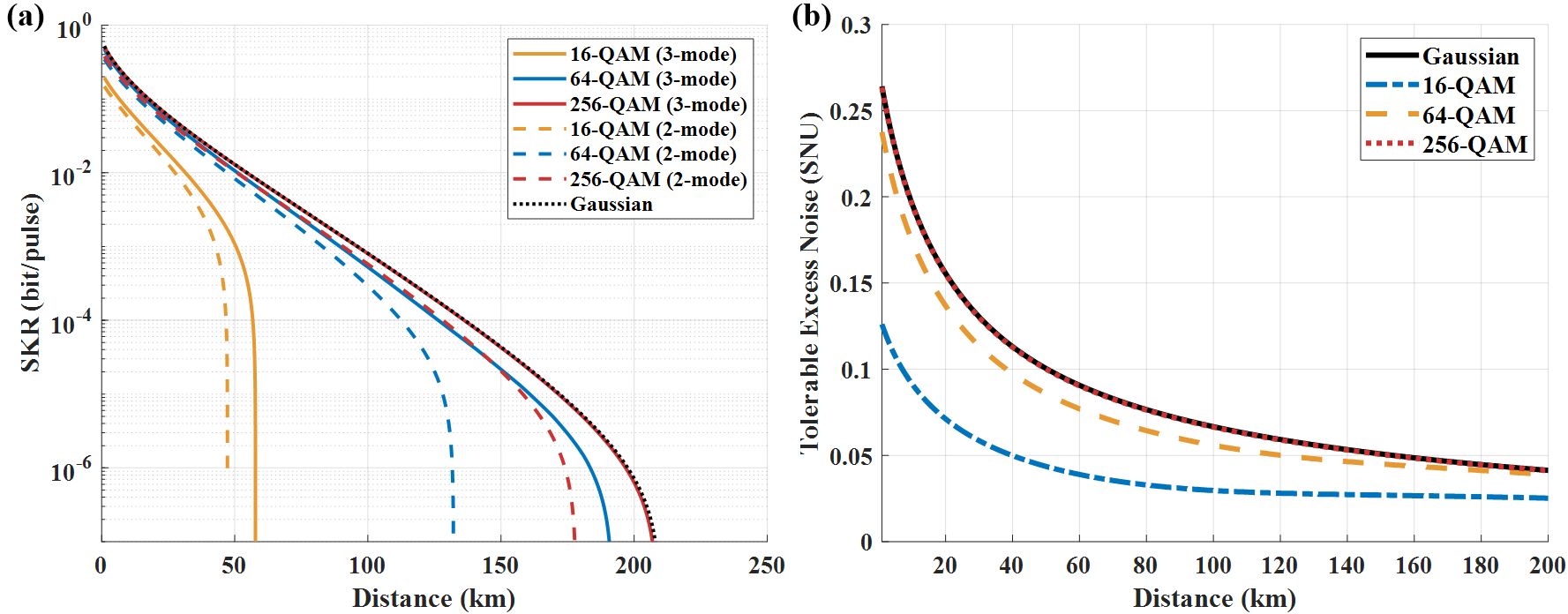}
    \caption{\label{Simulation1}
    (a) SKR of the DM protocol based on 3-mode matrix $\gamma_{ACB}^{\mathrm{sym}}$ and 2-mode matrix $\gamma_{AB}^{\mathrm{sym}}$ compared with the ideal GM protocol. The yellow, blue and red lines represent the 16-, 64-, and 256-QAM schemes, respectively. The black dotted line represents the ideal Gaussian case. (b) The tolerable excess noise of the DM protocol based on 3-mode matrix $\gamma_{ACB}^{\mathbf{sym}}$ with $\eta_{BS}=0.9$ compared to ideal Gaussian-modulated scheme. 
    Unified simulation parameters: reconciliation efficiency $\beta=0.95$, excess noise \(\varepsilon=0.04\), detection efficiency $\eta=0.6$, electronic noise $\nu_{ele}=0.1$ SNU. Mode variance $V_A=3$ SNU for 16-QAM and $V_A=5$ SNU for 64- and 256-QAM.
    }
\end{figure*}

\subsection{Compatibility with trusted detector modeling}\label{Compatibility}
To demonstrate the applicability of the proposed symmetrization framework in practical CV-QKD systems, we show that the symmetrization procedure developed for the original three-mode covariance matrix can be naturally extended to the detector-modeling scenario. In particular, the symmetrization of the three-mode state can be consistently generalized to the corresponding five-mode covariance matrix arising from homodyne or heterodyne detection.

The protocol with detector modeling is illustrated in Fig. \ref{Detector}. For heterodyne detection, we focus on the intermediate-frequency system, in which the \(x\) and \(p\) quadratures can be simultaneously obtained by down-converting the outputs of one balanced homodyne detector. In this case, the heterodyne detector can be considered to consist of two symmetrical homodyne detectors with the same detection efficiency $\eta$ and electronic noise $v_{el}$. Therefore, with the detector modeling, the original three-mode system \(ACB_1\) is extended to the five-mode system \(ACBFG\), and the covariance matrix can be written as 

\begin{equation}
    \begin{aligned}
      &\gamma_{ACBFG}^{(0)} =  \\
      & \begin{pmatrix}
V_A I_2 & \phi_{AC}\sigma_Z & \kappa_{AB} & \kappa_{AF} & 0 \\
\phi_{AC}\sigma_Z & V_C I_2 & \psi_{CB} & \psi_{CF} & 0 \\
\kappa_{AB}^T & \psi_{CB}^T & \gamma_B & \psi_{BF} & \phi_{BG}\sigma_Z \\
\kappa_{AF}^T
&
\psi_{CF}^T
&
\psi_{BF}^T
&
\gamma_F
&
\phi_{FG}\sigma_Z
\\
0
&
0
&
\phi_{BG}\sigma_Z
&
\phi_{FG}\sigma_Z
&
V_G I_2
\end{pmatrix}
    \end{aligned},
\label{eq:gamma_ACBFG}
\end{equation}
where the asymmetric sub-matrices are 
\begin{equation}
    \begin{aligned}
        &\kappa_{AB}=\sqrt{\eta}\kappa_{AB_1}=\sqrt{\eta}
    \begin{pmatrix}
        \kappa_{11} & \kappa_{12} \\
        \kappa_{21} & \kappa_{22}
    \end{pmatrix} \\
    &\kappa_{AF}=\sqrt{1-\eta}
    \begin{pmatrix}
        \kappa_{11} & \kappa_{12} \\
        \kappa_{21} & \kappa_{22}
    \end{pmatrix}\\
        &\psi_{CB}=\sqrt{\eta}\psi_{CB_1}=\sqrt{\eta}
    \begin{pmatrix}
        \phi_{11} & \phi_{12} \\
        \phi_{21} & \phi_{22}
    \end{pmatrix} \\
    &\psi_{CF}=\sqrt{1-\eta}
    \begin{pmatrix}
        \phi_{11} & \phi_{12} \\
        \phi_{21} & \phi_{22}
    \end{pmatrix}\\
        &\gamma_{B}=\eta
    \begin{pmatrix}
        b_{11} & b_{12} \\
        b_{21} & b_{22}
    \end{pmatrix}+(1-\eta)V_G I_2 \\
    &\gamma_{F}=(1-\eta)
    \begin{pmatrix}
        b_{11} & b_{12} \\
        b_{21} & b_{22}
    \end{pmatrix}+\eta V_G I_2 \\
    &\psi_{CF}=\sqrt{\eta(1-\eta)}\left(\begin{pmatrix}
        b_{11} & b_{12} \\
        b_{21} & b_{22}
    \end{pmatrix}-V_GI_2\right)
    \end{aligned},
\end{equation}
where $V_G=1+v_{el}/(1-\eta)$ for homodyne detection, and $V_G=1+2v_{el}/(1-\eta)$ for heterodyne detection.

Now we show that the five-mode covariance matrix can be symmetrized by the modified three-step symmetrization procedure. Since all blocks involving mode \(G\) are already symmetric, they remain unchanged throughout the symmetrization process and will not be discussed further.

In the first step, the controlled \(90^\circ\) phase-space rotation is modified as follows: conditioned on the value of the control state, mode \(A\) is randomly rotated by \(90^\circ\) clockwise, while modes \(C\), \(B\), \(F\), and \(G\) are simultaneously rotated by \(90^\circ\) counterclockwise. After the rotation, the covariance matrices \(\gamma_B\) and \(\gamma_F\) are simultaneously symmetrized. Meanwhile, the correlation blocks \(\phi_{CB}\) and \(\phi_{CF}\) have identical diagonal entries and opposite off-diagonal entries.

In the second step, a local rotation is applied to mode \(B\). Since the covariance matrices \(\phi_{CB}\) and \(\phi_{CF}\) differ only by an overall scaling factor, they share the same rotation angle. Consequently, the same local rotation simultaneously diagonalizes both correlation blocks.
The resulting covariance matrix can therefore be written as
\begin{equation}
    \begin{aligned}
      &\gamma_{ACBFG}^{(2)} =  \\
      & \begin{pmatrix}
V_A I_2 & \phi_{AC}\sigma_Z & \sqrt{\eta_d}K & \sqrt{1-\eta_d}K & 0 \\
\phi_{AC}\sigma_Z & V_C I_2 & \phi_{CB}I_2 & \phi_{CF}I_2 & 0 \\
\sqrt{\eta_d}K & \phi_{CB}I_2 & V_B I_2 & \phi_{BF}I_2 & \phi_{BG}\sigma_Z \\
\sqrt{1-\eta_d}K
&
\phi_{CF}I_2
&
\phi_{BF}I_2
&
V_F I_2
&
\phi_{FG}\sigma_Z
\\
0
&
0
&
\phi_{BG}\sigma_Z
&
\phi_{FG}\sigma_Z
&
v I_2
\end{pmatrix}
    \end{aligned},
\label{eq:gamma_ACBFG_sym}
\end{equation}
where $K=\begin{pmatrix}
    \kappa & \Delta \kappa \\
    \Delta \kappa & -\kappa
\end{pmatrix}$.

In the third step, the two states with covariance matrices corresponding to opposite values of \(\Delta\kappa\) are mixed with equal probabilities. Direct calculation shows that these two covariance matrices possess identical symplectic eigenvalues and therefore yield the same SKR. Consequently, the mixture satisfies the concavity requirement of the SKR, and the SKR scaling argument established for the three-mode case remains valid.

After these three steps, the five-mode covariance matrix is transformed into the symmetric form as
\begin{equation}
    \begin{aligned}
      &\gamma_{ACBFG}^{\mathrm{sym}} =  \\
      & \begin{pmatrix}
V_A I_2 & \phi_{AC}\sigma_Z & \sqrt{\eta_d}\kappa\sigma_Z & \sqrt{1-\eta_d}\kappa\sigma_Z & 0 \\
\phi_{AC}\sigma_Z & V_C I_2 & \phi_{CB}I_2 & \phi_{CF}I_2 & 0 \\
\sqrt{\eta_d}\kappa\sigma_Z & \phi_{CB}I_2 & V_B I_2 & \phi_{BF}I_2 & \phi_{BG}\sigma_Z \\
\sqrt{1-\eta_d}\kappa\sigma_Z
&
\phi_{CF}I_2
&
\phi_{BF}I_2
&
V_F I_2
&
\phi_{FG}\sigma_Z
\\
0
&
0
&
\phi_{BG}\sigma_Z
&
\phi_{FG}\sigma_Z
&
v I_2
\end{pmatrix}
    \end{aligned},
\label{eq:gamma_ACBFG_sym}
\end{equation}
Furthermore, the resulting symmetrized covariance matrix is identical to that obtained by first applying the symmetrization procedure to the original three-mode covariance matrix and then extending it to the five-mode matrix with detector modeling. Therefore, the detector model neither alters the symmetrization framework nor affects the SKR scaling relations derived in the previous section. The detector parameters, including the detection efficiency and the electronic noise, remain unchanged throughout the symmetrization procedure.






\begin{figure*}[t]
    \includegraphics[width=17 cm]{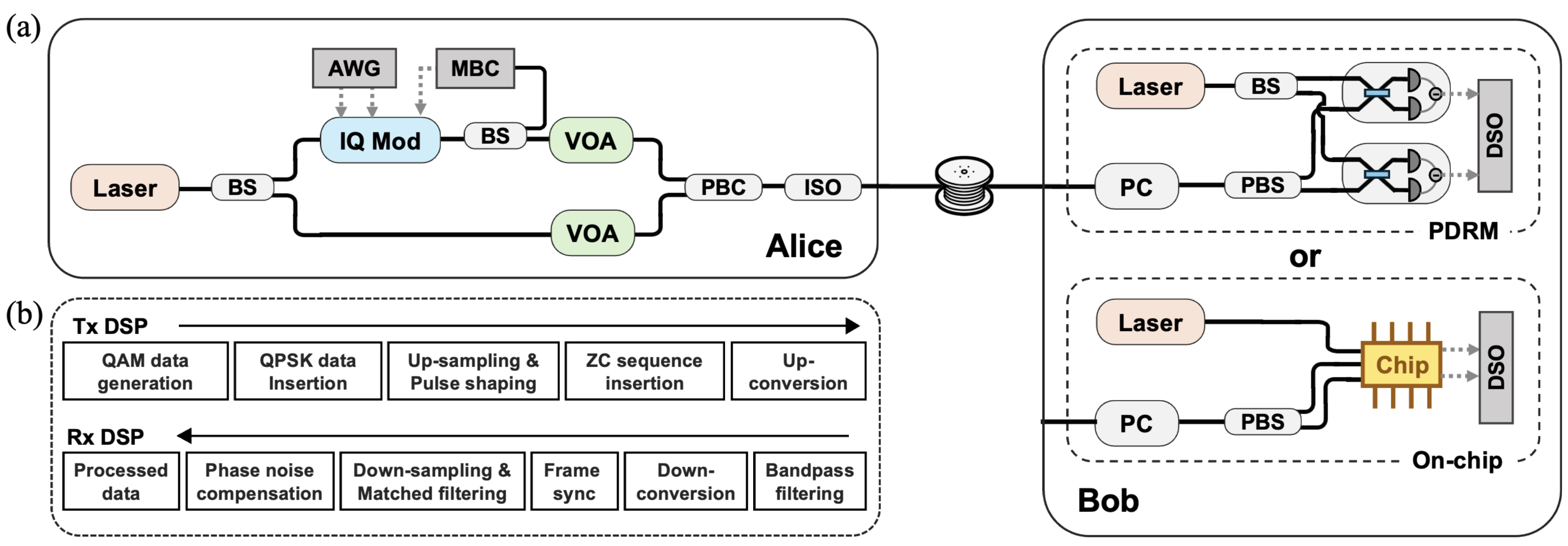}
    \caption{Experimental demonstration of the proposed discrete modulated CV-QKD protocol. (a) System layout. (b) DSP routine. BS, beam splitter; IQ Mod, in-phase/quadrature modulator; VOA, variable optical attenuator; PBC, polarization beam combiner; ISO, optical isolator; PC, polarization controller; PBS, polarization beam splitter; PDRM, polarization diversity receiver module; AWG, arbitrary waveform generator;  DSO, digital storage oscilloscope.\label{Fig: Exp_Scheme}
    }
\end{figure*}

\subsection{Simulations}\label{Simulation}
We present some additional simulation results for most commonly used 16-, 64-, and 256-QAM modulation schemes in this subsection. 
We adopt the parameter pairs $(V_G, V_A, \eta_A^{\mathrm{opt}})$ listed in Tab. \ref{tab1}. Considering imperfect error correction, the reconciliation efficiency is set to \(\beta=0.95\). Additionally, The detector efficiency and electronic noise are set to \(\eta_{D}=0.6\) and \(v_{\mathrm{el}}=0.1\).

The SKRs for three constellation cases are shown in Fig. ~\ref{Simulation1} (a) and compared with the ideal GM benchmark. The 256-QAM scheme nearly reaches the ideal Gaussian performance, while the 64-QAM scheme achieves comparable results, whereas the 16-QAM scheme exhibits noticeable degradation. These results demonstrate the effectiveness of our protocol with a high-order constellation. 
For the simplified 2-mode scheme, the SKR experiences a slight decrease, which can be largely compensated by increasing $\eta_{BS}$, as discussed in the main text.  

Figure~\ref{Simulation1} (b) shows the maximum tolerable excess noise for each modulation scheme. These results demonstrate the protocol's strong noise resilience: by selecting a suitable source, the protocol can tolerate at least \(0.05\) SNU of excess noise within 100~km, supporting practical commercial implementation and large-scale metropolitan deployment.

\section{Experimental demonstration}
The proposed protocol has been experimentally demonstrated on both the discrete-module-based and chip-based platforms. The experimental setup of the system is presented in Fig. \ref{Fig: Exp_Scheme} (a) and the DSP routine is illustrated in Fig. \ref{Fig: Exp_Scheme} (b).

\subsection{Transmitter}
Alice employs a continuous-wave (CW) laser (NKT Photonics Basik X15) operating at 1550.12 nm with a 0.1 kHz linewidth as the optical carrier. The laser output is split by a beam splitter into two branches: one undergoes QAM modulation, while the other serves as a pilot tone to provide a frequency and phase reference.

The QAM modulation is implemented using an in-phase/quadrature (IQ) modulator (Fujitsu FTM7962EP) driven by an arbitrary waveform generator (AWG) at a sampling rate of 30 GSa/s. With a bandwidth of 23 GHz, the IQ modulator supports a system baud rate of 1 GBaud. To ensure stability, a portion of the modulated signal is directed to a bias controller that maintains the modulator in carrier-suppression mode. Subsequently, a variable optical attenuator (VOA, EXFO FTBx-3500-BI) reduces the signal to the quantum level, achieving an average photon number consistent with the target modulation variance. This VOA integrates a real-time power monitor for precise calibration of the output power.

Finally, the modulated quantum signals and pilot tones are merged by a polarization beam combiner (PBC). They are aligned in orthogonal polarization directions and co-propagated through the untrusted quantum channel.

\begin{figure*}[t]
    \includegraphics[width=17 cm]{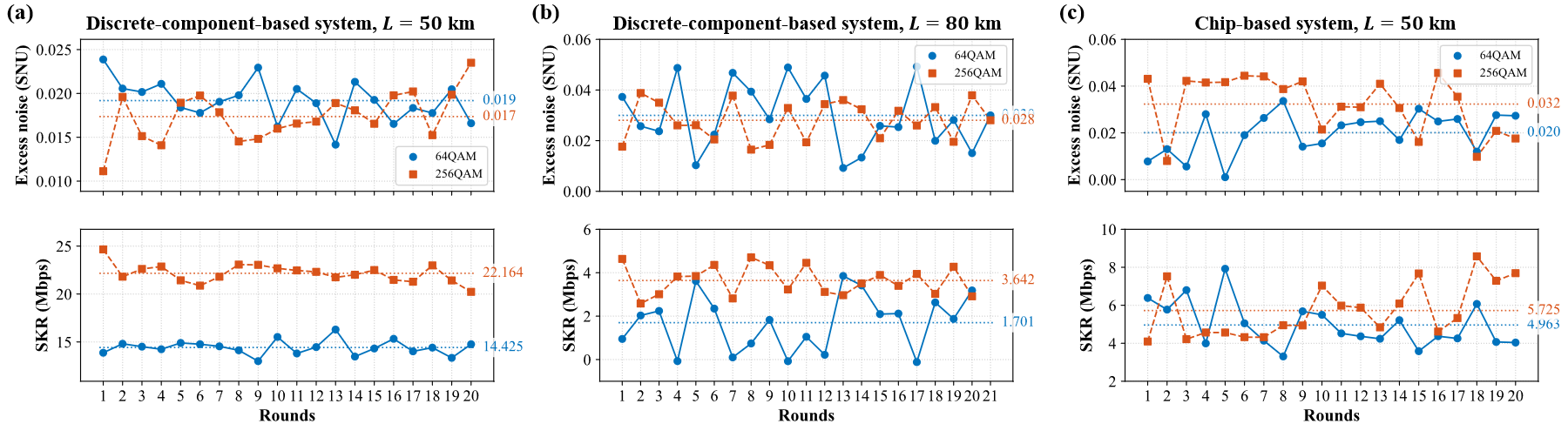}
    \caption{\label{Exp_results}Experimental results for both the discrete-module-based system and chip-based system with 64-QAM and 256-QAM modulation. \textbf{(a)} The estimated excess noise and secret key rates for discrete-component-based system at 50 km. Dashed lines denote the average values, which are labeled beside the corresponding curves. \textbf{(b)} The estimated excess noise and secret key rates for discrete-component-based system at 80 km. \textbf{(c)} The estimated excess noise and secret key rates for chip-based system at 50 km. The block size of each round is $2.56 \times 10^8$. Dashed lines denote the average values, which are labeled beside the corresponding curves.
    }
\end{figure*}

\subsection{Receiver}
The receiver, Bob, is designed to support two distinct configurations for shot-noise-limited coherent detection: a discrete Polarization-Diversity Receiver Module (PDRM) and an integrated on-chip solution.

In the PDRM configuration, the incoming signal from the untrusted quantum channel passes through a polarization controller (PC) followed by a polarization beam splitter (PBS) to demultiplex the orthogonal polarization components (the quantum signal and the pilot tone). A local oscillator (LO), provided by a dedicated laser (NKT Photonics Basik X15), is split by a 1:2 beam splitter (BS). The LO center frequency is offset by approximately 1.55 GHz relative to Alice’s transmitter laser to enable intermediate-frequency (IF) coherent detection. The resulting signals are then captured and digitized by a digital storage oscilloscope (DSO) for further analysis.

Alternatively, Bob can be implemented using a silicon photonic integrated chip. In this architecture, the signal (after the PC and PBS) and the LO laser are coupled directly into a silicon photonic chip, which performs highly balanced interference and shot-noise-limited detection within a compact footprint. The output electrical signals are similarly processed by the DSO. Compared to discrete fiber-based components, this integrated approach significantly minimizes the system footprint while enhancing scalability and mechanical stability.


\subsection{Digital signal processing}
The system performance is ensured by a comprehensive DSP routine, categorized into transmitter-side (Tx) and receiver-side (Rx) processing.

The Tx DSP prepares the driven signal for modulation through the following sequential stages. QAM data generation: random bit strings are mapped to a target PCS QAM constellation; QPSK data Insertion: training symbols, modulated using Quadrature Phase Shift Keying (QPSK) format, are multiplexed with the quantum data in time domain to facilitate phase estimation; Up-sampling and pulse shaping: The digital symbols are oversampled and filtered (using a root-raised-cosine filter) to limit the signal bandwidth and minimize inter-symbol interference (ISI); ZC sequence insertion: Zadoff-Chu (ZC) sequences are integrated into the frame structure to enable robust frame synchronization and time recovery at the receiver; Up-conversion: the baseband signal is digitally up-converted to an intermediate frequency (IF) to prepare it for physical modulation and transmission.

Upon detection, the digitized raw data undergoes a reverse processing chain to recover the information. Bandpass filtering: out-of-band noise is suppressed to improve the SNR of the captured waveform; Down-conversion: the IF signal is shifted back to baseband, extracting the complex envelopes of the modulated data; Frame synchronization: using the pre-inserted ZC sequences, the receiver identifies the start of each data frame with high temporal precision; Down-sampling and matched filtering: the signal is down sampled to the symbol rate, and matched filtering is applied to maximize the SNR and eliminate residual ISI; Phase noise compensation: leveraging the QPSK signals, the DSP compensates for the relative phase drift and frequency offset between the transmitter and LO lasers; Processed data: the final output consists of the recovered receiver variables, ready for post-processing steps such as parameter estimation, information reconciliation, and privacy amplification.

\subsection{Experimental results}

The experimental results for both the discrete-component-based system and chip-based system with 64-QAM and 256-QAM are shown in Fig. \ref{Exp_results}. The estimated excess noise and secret key rates under different links are presented respectively. Similarly, the mode variance $V_A=5$ SNU, which corresponds to a modulation variance of $V_M=3.6$ SNU. The variances of the Gaussian probabilistic constellation shaping are $V_G=6$ for 64-QAM and $V_G=11$ for 256-QAM. The training sequence ratio is 1/8. Furthermore, the reconciliation efficiency is set to $92\%$ for 64-QAM and $95\%$ for 256-QAM. The data block size used for parameter estimation in each test round was $2.56 \times 10^8$.
For each scenario, we conducted 20 rounds of testing. The results of all 20 rounds are plotted in Fig. \ref{Exp_results}, where the upper figures show the estimated excess noise, and the lower ones show the SKR. The average results of the 20-round tests are labeled on the right side of the graph. 
For the discrete-module-based system at 50 km (Fig. \ref{Exp_results} (a)), the average estimated excess noise is 0.019 SNU and 0.017 SNU for 64-QAM and 256-QAM, and the average SKRs are 14.425 Mbps and 22.164 Mbps for 64-QAM and 256-QAM respectively. 
For the discrete-module-based system at 80 km (Fig. \ref{Exp_results} (b)), the average estimated excess noise is 0.030 SNU and 0.028 SNU, while the SKRs are 1.701 Mbps and 3.642 Mbps respectively. For the chip-based system at 50 km (Fig. \ref{Exp_results} (c)), the average estimated excess noise is 0.032 SNU and 0.020 SNU, while the SKRs are 4.963 Mbps and 5.725 Mbps respectively.
The results show that the SKR under 256-QAM modulation is consistently higher than that of 64-QAM, and its excess noise exhibits smaller fluctuations. Increasing the modulation order generally improves the achievable SKR, a trend observed across all experiment configurations. Therefore, adopting higher-order modulation is an effective approach to improve the signal-to-noise ratio and approach the performance of an ideal Gaussian modulation scheme.

\end{document}